\documentclass[12pt, letterpaper]{article}

\usepackage{amsmath,amsthm,amssymb}
\usepackage[margin=1in]{geometry}
\usepackage{setspace}
\usepackage{natbib}
\usepackage{graphicx}
\usepackage{comment}
\usepackage{mathrsfs}
\usepackage{xcolor}
\usepackage{enumitem}

\title{Defensive Model Expansion for Robust Bayesian Inference}
\date{}
\author{Antonio R. Linero\thanks{Department of Statistics and Data Sciences,
    University of Texas at Austin, Austin TX, USA
  \\ Address for correspondence: Welch 5.216, 105 E 24th St D9800, 78705, Austin
  TX, USA. \texttt{antonio.linero@austin.utexas.edu}}}

\newcommand{\bmu}{\boldsymbol \mu}
\newcommand{\bK}{\boldsymbol K}
\newcommand{\bX}{\boldsymbol X}

\newcommand{\bZ}{\boldsymbol Z}
\newcommand{\Bernoulli}{\operatorname{Bernoulli}}

\newcommand{\Cov}{\operatorname{Cov}}
\newcommand{\E}{\mathbb E}

\newcommand{\Fisher}{\mathcal I}
\newcommand{\GP}{\operatorname{GP}}
\newcommand{\Holder}{Hölder}
\newcommand{\Identity}{\mathrm I}
\newcommand{\iid}{\stackrel{\text{iid}}{\sim}}
\newcommand{\Ihat}{\widehat I_N}

\newcommand{\InvGam}{\operatorname{InvGam}}
\newcommand{\ip}[2]{\langle #1 , #2 \rangle}

\newcommand{\Matern}{Matérn}

\newcommand{\Normal}{\operatorname{Normal}}
\newcommand{\Reals}{\mathbb R}

\newcommand{\sF}{\mathcal F}
\newcommand{\sH}{\mathcal H}
\newcommand{\sM}{\mathcal M}
\newcommand{\sN}{\mathcal N}
\newcommand{\sZ}{\mathcal Z}

\newcommand{\Tree}{\mathcal T}
\newcommand{\Var}{\operatorname{Var}}
\newcommand{\zeros}{\mathbf 0}

\newtheorem{theorem}{Theorem}

\newtheorem{corollary}{Corollary}

\theoremstyle{definition}
\newtheorem{definition}{Definition}
\newtheorem{remark}{Remark}

\usepackage[colorlinks, citecolor=blue]{hyperref}
  
\begin{document}

\maketitle

\begin{abstract}

  Some applied researchers hesitate to use nonparametric methods, worrying that they will lose power in small samples or overfit the data when simpler models are sufficient. We argue that at least some of these concerns are unfounded when nonparametric models are strongly shrunk toward parametric submodels. We consider expanding a parametric model with a nonparametric component $r(x)$ that is heavily shrunk toward zero. This construction allows the model to adapt automatically: if the parametric model is correct, the nonparametric component disappears, recovering parametric efficiency, while if it is misspecified, the flexible component activates to capture the missing signal. We show that this adaptive behavior follows from simple and general conditions. Specifically, we prove that Bayesian nonparametric models anchored to linear regression, including variants of Gaussian process regression and Bayesian additive regression trees, consistently identify the correct parametric submodel when it holds and give asymptotically efficient inference for regression coefficients. In simulations, we find that the \emph{general BART} model performs identically to correctly specified linear regression when the parametric model holds, and substantially outperforms it when nonlinear effects are present. This suggests a practical paradigm: \emph{defensive model expansion} as a safeguard against model misspecification.

\end{abstract}

\doublespacing

\section{Introduction}

In our experience, non-statistician collaborators are often cautious about using nonparametric methods. Many of their reasons for skepticism are reasonable: ``black-box'' procedures based on machine learning can be hard to interpret and, even when they improve predictions, may not lead to clear scientific understanding or actionable policy recommendations. Another common concern is that nonparametric methods may not justify the extra computational effort or expertise they require, relative to their impact on final conclusions.

In our view, however, some criticisms of nonparametric methods are less well-founded. One persistent worry is that nonparametric methods lose power in small samples, making it risky to use them out of fear that a ``statistically significant'' result might disappear. This need not be true: to the extent that they reduce error variability, nonparametric methods can actually increase power. And even if some power were lost, one could argue that including a nonparametric component provides a more honest reflection of modeling uncertainty. Another frequent complaint is that nonparametric models tend to overfit, and that parsimonious parametric models are safer. While overfitting can occur with poorly regularized nonparametric models, it is not an inherent problem.

Our goal is to show that properly constructed Bayesian nonparametric methods can perform just as well as parametric methods, even in small samples with high signal-to-noise ratios, by using a strategy of \emph{defensive model expansion}. The main idea is to add a flexible adjustment term $r(x)$, with shrinkage toward zero, such that when $r(x) \equiv 0$ we recover the original parametric model. This approach has several advantages:

\begin{itemize}

\item
  By incorporating model selection into the modeling process itself, we avoid relying on an analyst’s ad hoc (and often data-dependent) choices about which interactions or nonlinearities to include. This reduces researcher degrees of freedom and guards against \emph{garden of forking paths} problems at the analysis stage \citep{gelman2013garden}.

\item
  If the parametric model is badly misspecified or omits important signals, expanding into a nonparametric model can improve power. With a flexible ``alternative'' model where $r(x) \ne 0$, we can also use the estimated size of $r(\cdot)$ as a diagnostic for parametric model misspecification: if $|r(x)|$ is large on average then this is evidence that the parametric submodel fits poorly.

\item
  If the nonparametric component turns out to be important, we obtain more honest uncertainty quantification and the chance for richer scientific discoveries. Moreover, if we specify a fully Bayesian model with sensible priors in advance, then we have still only ``used the data once'' in performing inference on these discoveries \citep{woody2021model}.

\item
  When the parametric model is correct, we lose little power. In fact, we will show that we can still obtain asymptotically efficient parameter estimation.

\end{itemize}

The downsides of this strategy are minimal, as we can always shrink the nonparametric component heavily toward the parametric submodel. Often the result will simply be to confirm that the nested parametric model is a good approximation to reality. While this may seem anti-climactic, it is still useful information. Occasionally, the strategy may be too deferential to the parametric submodel, but still detect obvious signs of misspecification; we illustrate this in Section~\ref{sec:applications}, where correctly specified parametric models achieve predictive performance that is matched by nonparametric models. In large samples, a common outcome is that we identify interesting ways to expand the nested parametric model in a Bayes-valid fashion, which we see in Section~\ref{sec:diagnosing}.

\subsection{Existing Literature}

It has long been recognized that one can design Bayesian models that adaptively become more complex when the data demand it. This idea is a core theme in the Bayesian nonparametric literature, although in practice much of the focus has been on clustering and latent feature models \citep{orbanz2011bayesian,gershman2012tutorial,ghahramani2007bayesian}. A difference from the settings we consider is that clustering and latent feature problems are \emph{evidently complex} (with the main question being only \emph{how} complex), whereas we focus on settings that are \emph{plausibly simple}, where a parametric model might in fact be sufficient. See \citet{li2020comparing} for an example of using a baseline nonparametric model as a tool for evaluating parametric models, as we do here.

We briefly describe some other works that have been particularly influential on our approach. \citet{neal1995bayesian} makes the point that, ``There is no provision in the Bayesian framework for changing the model or the prior depending on how much data was collected. If the model and prior are correct for a thousand observations, they are correct for ten observations as well (though the impact of using an incorrect prior might be more serious with fewer observations),'' and that when the underlying phenomenon is evidently complex the appropriate solution to finding a simpler model outperforming a more complex model is ``to design a different complex model that captures whatever aspect of the problem led to the simple model performing well.''

In the context of mixture models, a selling point of infinite mixtures $f(y) = \sum_{k = 1}^\infty \omega_k \, g(y \mid \vartheta_k)$ is that they can collapse to a single parametric model $g(y \mid \theta)$ when most prior mass is placed on the event $\omega_1 \approx 1$. For Dirichlet process mixtures \citep{escobar1995bayesian}, this is accomplished by shrinking the concentration parameter heavily toward zero. For conditional mixture models --- such as dependent Dirichlet processes (DDPs, \citealp{maceachern2000dependent,quintana2022dependent}) --- this allows us to recover standard parametric generalized linear models when appropriate \citep{shahbaba2009nonlinear,hannah2011dirichlet}.

Our strategy also represents a type of \emph{continuous model expansion} in the sense of \citet{gelman2013philosophy}, who recommend ``forming a larger model that includes both A and B as special cases,'' as an alternative to testing Model A against Model B. Like them, we do not recommend formally testing the parametric submodel (except as a theoretical device). While they explicitly do not recommend Bayesian nonparametrics as a tool for continuous model expansion, we also do not feel that our proposal is at odds with the additional use of model checking they recommend, and do not view our use of an additional nonparametric component to be a complete solution to model misspecification. Related works include \citet{kamary2018testing} and \citet{yao2018using}.

\subsection{Outline}

We recommend defensive model expansion partly because, under very general conditions, it yields efficient inference when the associated parametric model is correctly specified. The story is particularly simple if we replace the posterior distribution with a \emph{fractional posterior} \citep{walker2001bayesian,o1995fractional}, which allows us to sidestep the technical testing and entropy conditions required by \citet{ghosal2000convergence}. In fact, the nonparametric component of the prior plays little direct role in our analysis, making the results extremely general. That being said, the main role played by the fractional posterior is to make establishing a posterior contraction rate simpler; the assumption that a fractional posterior is used can often be replaced with directly assuming a posterior concentration rate, for which there are many existing rules for the GP and BART models we consider in this work.


We begin by presenting general results on posterior concentration for nonparametric models when parametric submodels have sufficient prior mass. These are straightforward extensions of existing results for fractional posteriors. We prove results for convergence rates, model selection, and asymptotic normality under the assumption that the parametric model is correctly specified. Our model selection results essentially show that (i) we can correctly select the parametric model and (ii) the posterior distribution satisfies a Bernstein-von Mises theorem.

We also prove two semiparametric Bernstein-von Mises theorems (Theorem~\ref{thm:sbvm} and Theorem~\ref{thm:logistic-bvm}) that apply in settings where model selection consistency is difficult to apply, which guarantee efficient posterior concentration of a \emph{projection parameter} $\theta^\star$ around its true value; when the submodel is correct, this projection parameter will coincide with the parameter of interest. Such projection parameters are often used to summarize posterior distributions \citep{woody2021model}. To the best of our knowledge, Theorem~\ref{thm:logistic-bvm} is not implied by existing results.

We apply these results to Gaussian process (GP) and Bayesian additive regression tree (BART) models. We show that BART models naturally satisfy the conditions required for a Bernstein-von Mises result (see also \citealp{rockova2020semi}), as do ``spike-and-GP'' priors. While we do not study this here, our semiparametric Bernstein-von Mises results also may be extended to projection parameters even when the anchoring parametric model is misspecified. These theoretical results are validated on simulated data.

In addition to being theoretically simple, defensive model expansion is easy to implement in practice. For example, because our convergence rate and model selection results apply directly to BART models, existing software for the \emph{general BART model} of \citet{tan2019bayesian} makes it straightforward to obtain parametric adaptivity with default prior specifications. BART models are especially attractive because, in addition to adapting to parametric structures, they also adjust to locally varying smoothness and the presence of low-order interactions \citep{linero2018abayesian,rovckova2020posterior,jeong2023art}.

\subsection{Notation}

We let $Z_1, Z_2, \ldots$ denote a sequence of random variables taking values in a set $\sZ \subseteq \Reals^Q$ and write $\bZ_N = (Z_1, \ldots, Z_N)$ for the data. We consider a model $\sF = \{F_\theta : \theta \in \Theta\}$ for some parameter space $\Theta$ such that the $F_\theta$'s are mutually dominated by some measure $dz$ with associated densities $f_\theta(z)$. We assume that $Z_i \iid f_{\theta_0}$ for some $\theta_0 \in \Theta$, although we will also briefly discuss the setting where the model is misspecified. For any function $g$, let $\E_\theta \{g(Z)\}$ and $\Var_\theta \{g(Z)\}$ denote expectation and variance under $f_\theta$.

We assume that $\theta$ decomposes as $(r, \eta)$ where $r \in \sH$ is infinite-dimensional and $\eta \in E \subseteq \Reals^P$ is finite-dimensional, such that $\Theta = \sH \times E$. The family $\sF_0 = \{F_\theta : \eta \in E \ \text{and} \ r = r_0\}$ represents a parametric submodel of $\sF$. In our examples, $r(\cdot)$ is a mapping $\Reals^D \to \Reals$ and $r_0(\cdot) \equiv 0$. Intuitively, $\eta$ corresponds to the usual parametric component (e.g., regression coefficients), while $r$ is a flexible nonparametric adjustment. We let $\theta_0 = (r_0, \eta_0)$.

We let $\Pi(\cdot)$ denote a prior on $\Theta$, $L(\theta) = \prod_i f_\theta(Z_i)$ the likelihood function, and $R(\theta) = L(\theta) / L(\theta_0)$ the likelihood ratio. The (fractional) posterior is defined by $\Pi_\alpha(A \mid \bZ_N) = \int_A L(\theta)^\alpha \ \Pi(d\theta) / \int L(\theta)^\alpha \ \Pi(d\theta)$ for $\alpha \in (0,1]$.

\begin{paragraph}{Product Prior Assumption (PPA)}
  Throughout, we mostly assume independent priors on $r$ and $\eta$, i.e., $\Pi = \Pi_r \times \Pi_\eta$, with $\Pi_\eta$ admitting a continuous density $\pi_\eta$ on $E$ such that $\pi_\eta(\eta_0) > 0$. We also assume that $\sH$ is a Hilbert space with inner product $\ip{\cdot}{\cdot}$, and that there exists a norm $\|\cdot\|_{\sH}$ and positive constants $(C(\theta_0), \epsilon)$ such that
  \begin{math}
    \max\{K(\theta_0 \| \theta), V(\theta_0 \| \theta)\} 
    \le C(\theta_0)\{\|\eta - \eta_0\|^2 + \|r - r_0\|_{\sH}^2\}
  \end{math}
  for all $\theta \in \Theta$ satisfying 
  $\|\eta - \eta_0\|^2 + \|r - r_0\|_{\sH}^2 \le \epsilon^2/C(\theta_0)$.
  Here $K(\theta_0 \| \theta)$ and $V(\theta_0 \| \theta)$ are Kullback–Leibler type divergences, as defined in Section~\ref{sec:theory}. Different examples where we apply PPA will use different norms and inner products.
\end{paragraph}

\vspace{1em}

We now describe three concrete models where our results apply. We will focus mainly on Model NR for clarity.

\begin{paragraph}{Nonparametric Regression (Model NR)}
  Let $Z_i = (X_i, Y_i)$ with
  \begin{math}
    [Y_i \mid X_i = x] \sim \Normal\{x^\top\beta_0 + r_0(x), \sigma_0^2\}, 
    \text{ } X_i \sim F_X,
  \end{math}
  where $\theta = (r, \beta, \sigma)$ and $X_i \in [0,1]^D$ for some $D$. Anchoring to the parametric submodel with $r_0(x) \equiv 0$ recovers ordinary linear regression. Our goals are to estimate the regression function $\mu_0(x) = r_0(x) + x^\top\beta_0$ and, ideally, recover $\beta_0$ efficiently.
\end{paragraph}

\begin{paragraph}{Nonparametric Classification (Model NC)}
  We adopt the same framework as Model NR, except $Y_i \in \{0,1\}$ and $\E_{\theta} (Y_i \mid X_i = x) = \mu(x) = (1 + \exp\{-x^\top\beta - r(x)\})^{-1}$ so that $[Y_i \mid X_i = x]$ is a logistic regression and $\theta = (r, \beta)$. Our goal again is to estimate $\mu_0(x) = \E_{\theta_0}(Y_i \mid X_i = x)$ and estimate $\beta_0$ efficiently.
\end{paragraph}


\begin{paragraph}{Density Regression (Model DR)}
  Similar to Model NR, assume $[Y_i \mid X_i = x] \sim \Normal(x^\top\beta_0, \sigma_0^2)$ with $\theta = (r, \beta, \sigma)$, but rather than using $r(x)$ to adjust the regression function we set
  \begin{math}
    f_\theta(y \mid x) 
    = \frac{\Phi\{r(x, y)\} \phi(y \mid x^\top\beta, \sigma)}
           {\int \Phi\{r(x, y)\} \phi(y \mid x^\top \beta, \sigma) \ dy}
  \end{math}
  where $\phi(y \mid \mu, \sigma)$ is the $\Normal(\mu, \sigma^2)$ density and $\Phi(\cdot)$ is the $\Normal(0,1)$ distribution function. Once again, $\theta_0 = (r_0, \beta_0, \sigma_0)$ with $r_0 \equiv 0$. When $r(x, y)$ is a Gaussian process, this model is a conditional version of the \emph{Gaussian process density sample} \citep{murray2008gaussian}, which is an extremely flexible model for a conditional density. It was shown by \citet{li2023adaptive} that when $r(\cdot)$ is given a smooth BART (SBART) prior, the posterior concentrates at a near-minimax rate about $f_{\theta_0}$ if $r_0(x, y)$ is \Holder-smooth and other assumptions are satisfied. \citet{linero2022bayesian} use this strategy in the context of Bayesian nonparametric survival analysis. Complementing the adaptive concentration results provided in those works, the results here also establish that these models can adapt to the underlying parametric regression model being correctly specified.
\end{paragraph}

\section{Modeling Strategies}

We now describe strategies for modeling the nonparametric component $r(\cdot)$ and for assessing how much the nested parametric submodel may be misspecified. We focus on two widely used classes of Bayesian models: Bayesian additive regression trees (BART) and Gaussian processes (GPs). 

\subsection{General Bayesian Additive Regression Trees Models}
\label{sec:general-bart}

For routine use, we recommend the Bayesian Additive Regression Trees (BART, \citealp{chipman2010bart}) prior for the nonparametric component $r(\cdot)$. In the context of Model NR, the regression function is $\mu(x) = x^\top\beta + \sum_{t = 1}^T g(x; \Tree_t, \sM_t)$ where each $g(x; \Tree_t, \sM_t)$ is a regression tree, $\Tree_t$ denotes the tree structure, and $\sM_t = \{\lambda_{\ell t} : \ell=1,\ldots,L_t\}$ is the set of leaf node predictions. \citet{tan2019bayesian} refer to this as a \emph{general BART} (GBART) model, and it can be fit using the \texttt{SoftBart} package \citep{linero2022softbart} using the \texttt{gsoftbart\_regression} function. A schematic of the construction of the model is given in Figure~\ref{fig:treefig}, which shows a single tree from the ensemble and the associated step function.
 
\begin{figure}[t]
  \centering
  \includegraphics[width=0.8\textwidth]{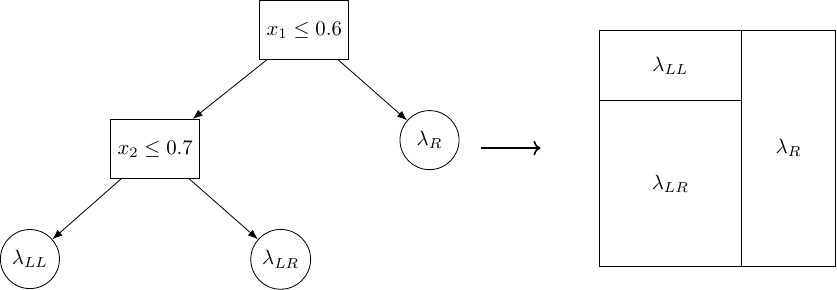}
  \caption{Left: an example of a regression tree with input $x = (x_1, x_2)$ supported on $[0,1]^2$. Right: the induced step function on $[0,1]^2$ \label{fig:treefig}.}
\end{figure}

We now review the prior specification of BART, which is important for our purposes. The default recommendation of \citet{chipman2010bart} takes $\lambda_{t\ell} \sim \Normal(0, \sigma^2_\mu / T)$ with $\sigma_\mu = 0.5$ after centering and scaling the $Y_i$'s to lie between -0.5 and 0.5.
The scaling of $\sigma^2_\mu$ by $T$ is done so that $\Var\{\sum_{t=1}^T g(x; \Tree_t, \sM_t)\} = \sigma^2_\mu$, which in turn leads to a Gaussian process limit as $T \to \infty$ \citep{linero2017review}. One way to shrink the GBART model to a linear regression is to place a prior on $\sigma_\mu$ supported near $0$, such as a half-normal or half-Cauchy prior.

Another way for GBART to reduce to a linear regression is for all of the trees in the ensemble to be ``empty'' in the sense of not splitting on any variables. In Section~\ref{sec:theory-bart} we will show that this suffices for GBART to possess strong theoretical properties when the linear model is correct. \citet{chipman2010bart} propose a branching process prior for $\Tree_t$ that (i) decides whether each node at a given depth of the tree will be converted into a branch node with two leaf children; (ii) iterates over the depths; and (iii) terminates if all of the nodes at a given depth are leaf nodes. The probability of a node at depth $d$ being a branch node under the prior is specified to be $q(d) = a / (1 + b)^d$ where by default $a = 0.95$ and $b = 2$. In particular, the probability that a tree consists of \emph{no splits} is $(1 - a)$, and the prior probability that the GBART model reduces to the linear model (with intercept offset by the constant $r(0)$) is therefore $(1 - a)^T > 0$.

\subsection{Gaussian Process Regression}
\label{sec:gaussian-process-regression}

A common approach for anchoring Model NR to a parametric submodel is to use the \emph{Gaussian process regression} (GP regression) model. GPs are distributions on random functions, and can be used as priors for regression functions $r(\cdot)$. A GP is determined by a mean function $m(x)$ and a positive semi-definite kernel $\kappa(x,x')$, with $\E_\Pi r(x) = m(x)$ and $\Cov_\Pi\{r(x), r(x')\} = \kappa(x, x')$. We write $\GP(m, \kappa)$ for the associated Gaussian process. For a review of Gaussian processes and alternative choices of $\kappa(\cdot,\cdot)$, see \citet{rasmussen2006gaussian}. Importantly, GPs are conjugate priors for $\mu(x)$ in the model $[Y_i \mid X_i = x] \sim \Normal\{\mu(x), \sigma^2\}$ when $\kappa(x, x')$ is held fixed.

For model NR, we take $r \sim \GP\{0, \sigma^2_\mu \, \kappa(\cdot, \cdot)\}$. Equivalently, letting $\mu(x) = x^\top\beta + r(x)$ we have $\mu \sim \GP(m, \sigma^2_\mu \, \kappa)$ where $m(x) = x^\top\beta$. If $\beta \sim \Normal(0, \sigma^2_\beta \Identity)$ this simplifies further to $\mu \sim \GP(0, \sigma^2_\mu \kappa + \sigma^2_\beta \kappa_{\text{linear}})$ where $\kappa_{\text{linear}}(x,x') = x^\top x'$. In our illustrations we use the squared exponential function $\kappa_{\text{SE}}(x, x') = \exp(-\rho \|x - x'\|^2)$, with the lengthscale $\rho$ controlling the smoothness of samples from the prior.

Just as with BART, there are multiple ways to obtain adaptation to parametric structures. Below, we describe some strategies for ensuring that a GP is capable of attaining adaptivity to $r_0(x) = 0$.

\paragraph{Scale Mixture Models}
One option is to place a hyperprior on $\sigma_\mu$, with heavy shrinkage toward zero. We call this the \emph{scale mixture of GPs} approach. Intuitively, this lets the model collapse to the parametric submodel when the nonparametric adjustment is unnecessary by taking $\sigma_\mu \approx 0$. When performing inference on projection parameters (see Section~\ref{sec:summarizing}) this strategy is sufficient for most purposes.

\paragraph{Spike and Slab Models}
A second option is to use a ``spike-and-slab'' prior \citep{george1993variable} with a GP slab distribution:
\begin{equation*}
  r \sim p_0 \, \delta_{r_0} + (1 - p_0) \, \GP(m, \kappa),
\end{equation*}
where $\delta_{r_0}$ denotes a point mass at $r_0(x) \equiv 0$ and $0 < p_0 < 1$. We refer to this as the \emph{spike-and-GP} approach. If $p_0$ is large then the model strongly favors the parametric submodel, while still allowing flexibility if supported by the data. In Section~\ref{sec:general-results} we will see that spike-and-slab priors adapt to parametric structures with essentially no further requirements. 

\paragraph{Semiparametric Bernstein-von Mises for Projections}
A third strategy is to simply use a GP for $r(x)$ and then perform inference on a projection parameter $\eta^\star$ obtained from projecting the full model onto the parametric submodel $\sF_0$. Concretely, we might compute the \emph{Kullback-Leibler projection} $\theta^\star = \arg \min_{\widetilde theta \in \sF_0} K(\theta \| \widetilde \theta)$ for each sample $\theta$ from the posterior and, because $\theta_0 \in \sF_0$, hopefully $\theta^\star$ will concentrate faster around $\theta_0$ than $\theta$. This can be used with either of the previous two approaches, but sometimes works on its own as well. When the parametric model is correctly specified, appropriately chosen GP priors can yield efficient inference. However, poorly chosen priors may still lead to inefficient or invalid inference. For example, we will see that even if the posterior converges at a near-parametric rate such as $N^{-1/2}(\log N)^t$ for some $t>0$, the bias in estimating $\beta_0$ may decay very slowly, making inference inefficient in finite samples. This occurs, for example, when we incorrectly set $\beta = 0$ and model $r(x)$ with a squared exponential kernel.

\subsection{Orthogonalization}

When using models with predictions of the form $X_i^\top \beta + r(X_i)$, such as Model NR or Model NC, it is helpful to require the nonparametric component $r(x)$ to be orthogonal to the linear predictors. That is, we require that $\sum_i r(X_i) \, X_i = \zeros$. Here this is meant primarily as an empirical, finite-sample orthogonalization device based on the observed design matrix rather than as a population-level identifiability assumption. This condition serves two purposes. First, it can substantially improve the mixing of MCMC algorithms used to fit the model. Second, it prevents $r(\cdot)$ from becoming colinear with $X_i^\top \beta$, making it easier to shrink $r(\cdot)$ toward zero when it is unnecessary.
Orthogonalization is not required for our theoretical results to hold, but it is often useful in practice. For GP models, orthogonalization is implemented by replacing the kernel $\kappa(x,x')$ with the projected kernel
\begin{math}
  \kappa^\star(x,x') = \kappa(x,x') - k_x^\top \bX (\bX^\top \bK \bX)^{-1} \bX^\top k_{x'},
\end{math}
where $\bK_{ij} = \kappa(X_i, X_j)$, $k_x = (\kappa(x, X_1), \ldots, \kappa(x, X_N))^\top$, and $\bX \in \Reals^{N \times P}$ is the design matrix \citep{plumlee2018orthogonal}.

\subsection{Summarizing Nonparametric Fits}
\label{sec:summarizing}

Our proposed workflow for using these models, illustrated in Section~\ref{sec:diagnosing}, is to:
\begin{enumerate}
\item Obtain posterior samples from the larger nonparametric model.
\item Project these samples onto the parametric submodel $\sF_0$ and assess its adequacy.
\item If the parametric model is inadequate with high posterior probability, fit additional summaries (such as additive models or tree-based fits) to identify specific aspects of the misspecification, and use these to refine $\sF_0$.
\end{enumerate}
One practical implementation is to form a linear projection summary by minimizing
\begin{math}
  \text{SSE} = \sum_{i=1}^N \{\mu(X_i) - X_i^\top \beta^\star\}^2.
\end{math}
We then compute
\begin{math}
  R^2 = 1 - \frac{\text{SSE}}{\sum_i \{\mu(X_i) - \bar \mu\}^2}, 
  \text{ where } 
  \bar{\mu} = N^{-1}\sum_i \mu(X_i).
\end{math}
We can interpret $R^2$ as the proportion of the overall signal that is explained by the parametric model. If $R^2$ is sufficiently close to 1, we conclude that the parametric model is adequate. The required closeness depends on context: if prediction is the main goal, we might demand $R^2$ very close to 1, while if the goal is only a broad scientific conclusion, a more modest $R^2$ could suffice. As a reporting convention, we recommend presenting the posterior mean and posterior distribution of $R^2$; for many descriptive purposes, values above roughly $0.99$ indicate an essentially interchangeable summary, values around $0.95$ indicate a close but imperfect summary, and materially smaller values suggest that richer structure may matter. If the linear projection is inadequate, we expand to an additive model $\mu(x) \approx \sum_j \gamma_j(x_j)$ and, if that is still insufficient, apply a CART fit to residuals to detect missing interaction terms. For logistic regression, it is natural instead to minimize the Kullback–Leibler divergence between the fitted probabilities $\mu(X_i)$ and a logistic regression model with predictor $X_i^\top \beta^\star$. 

Importantly, we do not recommend formal hypothesis testing of whether $r = r_0$. First, we may lack power to detect $r = r_0$ even if the model technically allows for it. Second, even when $r \ne r_0$ with high posterior probability, the parametric submodel may still be practically indistinguishable from the nonparametric model in terms of its predictions; this can occur in the GBART model, for example, when a shrinkage prior is placed on $\sigma_\mu$.

A criticism of the above strategy is that it is not specific enough, requiring subjective judgments about whether a summary is ``adequate'' or not. On the contrary, we feel that this is a strength, as it forces the analyst to engage with the question of what purpose the fitted model is supposed to serve and whether it is sufficient for that purpose. Judgments, such as what an appropriate tradeoff between interpretability and predictive accuracy are, should chiefly be determined in collaboration with stakeholders.

\section{Theoretical Results}
\label{sec:theory}

We now establish some basic theoretical properties for the BART and GP models. The results in this section split into two parts: general fractional-posterior concentration and model-selection results under correct parametric specification, and ordinary-posterior semiparametric Bernstein-von Mises results for projection parameters. To formalize the idea that a properly anchored Bayesian nonparametric model can perform ``just as well'' as a parametric submodel, we will use arguments based on the \emph{rate of convergence} of a posterior distribution; the rate can be used as a justification on its own, or as part of the assumptions for stronger justifications such as model selection consistency or a Bernstein-von Mises theorem. We will say that the posterior convergence rate is (upper bounded by) $\epsilon_N$ if we have $\Pi_\alpha\{H(\theta_0, \theta) > M \, \epsilon_N \mid \bZ_N\} \to 0$ in probability for some $M > 0$. Here, $H(\theta_0, \theta) = [\int \{\sqrt{f_\theta(z)} - \sqrt{f_{\theta_0}(z)}\}^2 \ dz]^{1/2}$ denotes the Hellinger distance between $f_\theta$ and $f_{\theta_0}$. For example, we will show that appropriate GBART and spike-and-GP priors give $\epsilon_N = N^{-1/2}$ rates (possibly up-to a logarithmic term) and that posterior inferences for $\beta_0$ are efficient, provided that the parametric submodel is correct.

The most important condition governing whether parametric adaptivity is possible is whether or not the prior assigns enough mass to neighborhoods of $\theta_0$. In nonparametric setups, the relevant neighborhoods are defined in terms of the Kullback-Leibler (KL) neighborhoods:
\begin{equation*}
  \theta \in B_\epsilon
  \quad \text{if} \quad
  K(\theta_0 \| \theta)
  \stackrel{\text{def}}{=}
  \E_{\theta_0} \log \frac{f_{\theta_0}(Z)}{f_\theta(Z)} \le \epsilon^2
  \quad \text{and} \quad
  V(\theta_0 \| \theta)
  \stackrel{\text{def}}{=}
  \Var_{\theta_0} \log \frac{f_{\theta_0}(Z)}{f_\theta(Z)} \le \epsilon^2.
\end{equation*}
We expect to be able to estimate $\theta_0$ at a rate $\epsilon_N$ if the following prior thickness condition holds:
\begin{definition}[$\epsilon_N$-thickness]
  The prior $\Pi$ is said to be $\epsilon_N$-thick at $\theta_0$ if there exist $C_1, C_2 > 0$ such that
  \[
    \Pi(B_{\epsilon_N}) \ge C_1 \, \exp(-C_2 N \epsilon_N^2).
  \]
\end{definition}
We simplify our arguments by restricting attention to the \emph{$\alpha$-fractional posterior} given by
\begin{equation}
  \label{eq:bayes}
  \Pi_\alpha(d\theta \mid \bZ_N)
  =
  \frac{L_N(\theta)^\alpha \ \Pi(d\theta)}{\int L_N(\theta)^\alpha \ \Pi(d\theta)},
  \qquad \alpha \in (0,1),
\end{equation}
where $L_N(\theta) = \prod_{i = 1}^N f_{\theta}(Z_i)$ is the likelihood; extending to $\alpha = 1$ (i.e., the genuine posterior) is possible for the models we consider here under additional conditions on the prior. Works that have studied the properties of fractional posteriors include \citet{o1995fractional, walker2001bayesian,bhattacharya2019bayesian,martin2019data}. Under essentially no further assumptions, the prior thickness condition is sufficient to guarantee posterior concentration at the rate $\epsilon_N$. Section 2.2 of \citet{castillo2024bayesian} is particularly instructive for understanding why these conditions are sufficient and how the arguments can be extended to non-iid settings.

\begin{paragraph}{Interpretation of $\alpha$}
  The role of $\alpha$ can be understood in two ways. First, it is as if we ``hold out'' $100(1-\alpha)\%$ of the data, which makes the posterior less prone to overfitting \citep{o1995fractional}. Second, as noted by \citet{walker2001bayesian}, the fractional posterior is equivalent to using the full likelihood $L_N(\theta)$ with a data-adaptive prior $\Pi^\star(d\theta) \propto \Pi(d\theta) / L_N(\theta)^{1-\alpha}$. This construction downweights parameter values that fit the observed data too well, again limiting overfitting. This robustness to overfitting is what allows us to avoid conditions on the size of the effective support used, for example, by \citet{ghosal2000convergence} to study the usual posterior.
\end{paragraph}


\begin{paragraph}{Non-iid Data and Finite Sample Bounds}
  It is possible to extend our results to allow for non-iid data by instead defining $B_\epsilon$ in terms of $\E_{\theta_0}  \left(\log \frac{f(Z_1, \ldots, Z_N \mid \theta_0)}{f(Z_1, \ldots, Z_N \mid \theta)}  \right)^k \le N \, \epsilon_N^2$ for $k \in \{1,2\}$, and similarly redefining the Hellinger distance. Another advantage of fractional posteriors is that it is also possible to convert convergence rate arguments into non-asymptotic arguments and allow for model misspecification, as done by \citet{bhattacharya2019bayesian}.
\end{paragraph}

\begin{paragraph}{Non-Fractional Posteriors}
  The usual posterior ($\alpha=1$) requires additional technical work that we feel distracts from the message of this paper. In this case, prior thickness alone is not enough; one must also control the complexity of the support of the prior, often by placing high mass on a low-entropy \emph{sieve} \citep{ghosal2000convergence}.
  While this strategy works, it places much stronger restrictions on the choice of $\Pi_r(dr)$. Nevertheless, we note that the specific semiparametric Bernstein-von Mises results given in Theorem~\ref{thm:bart-bvm}, Theorem~\ref{thm:sbvm}, and Theorem~\ref{thm:logistic-bvm} happen to be true for genuine posteriors; the fractional posterior only makes establishing the required rate of convergence easier, but convergence rate results for the non-fractional posterior are widely available for BART and GP models.
\end{paragraph}

\subsection{General Results}
\label{sec:general-results}

We start with general theoretical results that will later be applied to BART and GP regression models. Throughout this section (and Sections~\ref{sec:theory-bart} and~\ref{sec:theory-gps}) we fix $\alpha \in (0,1)$. Our first result is a general fractional posterior concentration theorem, implied for example by Theorem~2.1 of \citet{castillo2024bayesian}. It establishes a near-parametric rate of convergence under essentially only a prior thickness condition.

\begin{theorem}[General Rates]
  \label{thm:general-1}
  Suppose that PPA holds and that there exist constants $C_1, C_2, N^\star$ such that $\Pi_r(\|r - r_0\|_{\sH} \le M_N N^{-1/2}) \ge C_1 e^{-C_2 M_N^2}$ where $M_N / \sqrt{\log N} \to \infty$. Then we have $\E_{\theta_0} \Pi_\alpha\{H(\theta_0, \theta) > K \, M_N \, N^{-1/2} \mid \bZ_N\} \to 0$ as $N \to \infty$ for some sufficiently large $K$.
\end{theorem}

It is easy to design priors that satisfy the conditions of Theorem~\ref{thm:general-1}. In fact, any prior that assigns positive mass to the event $[r = r_0]$ will suffice. Spike-and-slab priors \citep{george1993variable} of the form $\Pi_r(dr) = p_0 \, \delta_{r_0}(dr) + (1 - p_0) \, \Pi^\star_r(dr)$ (with $\delta_{r_0}$ denoting a point mass distribution at $r_0$) are natural choices.

\begin{corollary}[Rate for Spike-and-Slab]
  \label{cor:general-1}
  Suppose that PPA holds, that $\Pi_r(r = r_0) > 0$, and that $M_N / \sqrt{\log N} \to \infty$. Then we have $\E_{\theta_0} \Pi_\alpha\{H(\theta_0, \theta) \ge M_N \, N^{-1/2} \mid \bZ_N\} \to 0$ as $N \to \infty$.
\end{corollary}

Anchoring to the parametric model also enables consistent model selection. If the parametric submodel holds and the nonparametric part cannot achieve a near-parametric rate on its own, then the posterior concentrates all its mass on the parametric model.  

\begin{theorem}[General Model Selection]
  \label{thm:general-consistency}
  Suppose that PPA holds and that $\Pi_r(dr) = p_0 \delta_{r_0}(dr) + (1 - p_0) \, \Pi_r^\star(dr)$ with $0 < p_0 < 1$. Define
  \begin{align*}
    \Pi^\star_\alpha(A \mid \bZ_N)
    = \frac{\int_A L(\theta)^\alpha \ \Pi_r^\star(dr) \ \Pi_\eta(d\eta)}
           {\int L(\theta)^\alpha \ \Pi_r^\star(dr) \ \Pi_\eta(d\eta)},
  \end{align*}
  and suppose that $\E_{\theta_0} \Pi^\star_\alpha\{H(\theta_0, \theta) < M_N \, N^{-1/2} \mid \bZ_N\} \to 0$ where $M_N / \sqrt{\log N} \to \infty$ as $N \to \infty$. Then $\E_{\theta_0} \Pi_\alpha(r = r_0 \mid \bZ_N) \to 1$ as $N \to \infty$.
\end{theorem}

\begin{remark}
  The assumption that $\Pi^\star_\alpha(\cdot\mid\bZ_N)$ fails to achieve a near-parametric rate is stronger than strictly necessary, as fractional Bayes factors are typically consistent for nested parametric model selection \citep{o1995fractional} (where the rate is $N^{-1/2}$ for both models). We discuss settings where it is obtained for GP priors in Section~\ref{sec:theory-gps}, where we note that some GP priors \emph{can} obtain near-parametric rates; these models fit more naturally with the posterior projection approach described in Section~\ref{sec:theory-gps} in terms of analysis, as a near-parametric rate is useful for proving a Bernstein-von Mises result. 
\end{remark}

When Theorem~\ref{thm:general-consistency} is combined with a parametric submodel $\sF_0$ that is differentiable in quadratic mean (see Section~5.5 of \citealp{vandervaart1998asymptotic}), we also obtain a Bernstein-von Mises result. The following theorem formalizes this under standard regularity conditions.

\begin{theorem}[Bernstein-von Mises]
  \label{thm:general-bvm}
  Suppose that the conditions of Theorem~\ref{thm:general-consistency} hold and that Assumption BVM in the Supplementary Material holds. Let $\Pi_\alpha(d\eta \mid \bZ_N)$ denote the marginal distribution of $\eta$ under $\Pi_\alpha(d\theta \mid \bZ_N)$. Then
  \begin{align*}
    \|\Pi_\alpha(d\eta \mid \bZ_N) - \Normal(d\eta \mid \widehat \eta, (\alpha N)^{-1} \Fisher(\eta_0)^{-1})\|_{TV}
    \to 0 \quad
    \text{in $F_{\theta_0}$-probability},
  \end{align*}
  where $\Fisher(\eta)$ denotes the unit Fisher information in the parametric model $\sF_0$, $\|\cdot\|_{TV}$ denotes total variation distance, and $\widehat \eta$ is any estimator such that $\sqrt N(\widehat \eta - \eta_0) \to \Normal(0, \Fisher(\eta_0)^{-1})$ in distribution.
\end{theorem}

The value of Theorem~\ref{thm:general-bvm} relative to the usual Bernstein-von Mises theorem is that it represents an oracle property of the fractional posterior: if the parametric model is correctly specified, our nonparametric prior will nevertheless obtain efficient inference for $\eta_0$. Efficiency holds up-to inflation of the variance by a factor of $\alpha^{-1}$; this can be corrected using the shift-and-rescale approach of \citet{l2023semiparametric}. On the other hand, if the parametric model is misspecified we will also be able to obtain efficient nonparametric estimation rates for $\theta_0$; the exact rates we get under misspecification will depend on the choice of $\Pi^\star(dr)$ and on $r_0(x)$.

\subsection{Results for Bayesian Additive Regression Trees}
\label{sec:theory-bart}

As illustrated by \citet{rovckova2020posterior,jeong2023art,linero2018abayesian}, BART models are adaptive to a wide variety of structures. In Model NR, for example, they are adaptive to sparsity, (locally-varying) smoothness, and low-order interaction structure in $r_0(x)$ out-of-the-box. Adding to this list of desirable properties, appropriately constructed BART models are also adaptive to parametric structures with no additional modifications.

\begin{theorem}[BART Adaptivity]
  \label{thm:bart-adaptivity}
  Suppose that Model NR holds with $r_0(x) = 0$ for all $x \in [0,1]^D$ and suppose that the prior $\Pi(d\theta)$ is such that 
  \begin{itemize}
  \item $(\beta, \sigma)$ has a continuous density $\pi_{\beta, \sigma}$ with $\pi_{\beta, \sigma}(\beta_0, \sigma_0) > 0$.
  \item $r$ is independent of $(\beta, \sigma)$ and has a BART prior as described in Section~\ref{sec:general-bart}.
  \end{itemize}
  Then the conditions of Theorem~\ref{thm:general-1} hold. Additionally, let $\sN$ be the event that all decision trees in the ensemble contain no splitting rules. Then $\E_0 \Pi_\alpha(\sN \mid \bZ_N) \to 1$.
\end{theorem}

This result is close in spirit to Theorem~\ref{thm:general-consistency}. While we cannot directly apply Theorem~\ref{thm:general-bvm}, Theorem~\ref{thm:bart-adaptivity} is strong enough to imply a Bernstein-von Mises result.  

\begin{theorem}[BART Bernstein-von Mises]
  \label{thm:bart-bvm}
  Suppose that the conditions of Theorem~\ref{thm:bart-adaptivity} hold and that $F_X$ is such that $X_{i1} \equiv 1$ and that $\E(X_i \, X_i^\top) = \Sigma_X$ is non-singular. Define the vector $\eta = (\beta_1 + r(0), \beta_2, \beta_3, \ldots, \beta_D, \sigma)^\top$ and let $\widehat \eta$ be the maximum likelihood estimator of $(\beta, \sigma)$ in Model NR when $r_0 \equiv 0$. Then the conclusion of Theorem~\ref{thm:general-bvm} holds.
\end{theorem}


\subsection{Results for Gaussian Processes and Posterior Projections}
\label{sec:theory-gps}

We have discussed two distinct strategies for Gaussian process models. The first is to use a spike-and-slab prior $\Pi(dr) = p_0 \, \delta_{r_0} + (1 - p_0) \, \Pi^\star_r(dr)$ with $\Pi^\star_r(dr)$ a (mixture of) Gaussian processes. The second is to use a scale mixture of Gaussian processes of the form $[r \mid \sigma_r] \sim \GP(0, \kappa)$ where $\kappa(x, x') = \sigma^2_r \, \rho(x,x')$ and $\rho(x,x')$ is a fixed kernel and $\sigma^2_r$ is given a hyperprior. For spike-and-slab priors, Corollary~\ref{cor:general-1} applies immediately. 


\begin{remark}[Gaussian Process Lower Bounds]
  Applying Theorem~\ref{thm:general-consistency} directly requires us to have a \emph{lower bound} on the rate of convergence for Gaussian process models. In other words, we need to verify that a GP prior without a spike-and-slab component does not achieve the near-parametric $(\log N / N)^{1/2}$ rate. We might expect flexible nonparametric methods to attain a rate of the form $\epsilon_N = N^{-p/(2p + D)}$ up-to logarithmic terms for some $p > 0$; this is the case for \Matern\ priors (see see Theorem~5 of \citealp{van2011information} and the discussion thereafter). For some common kernels (such as the squared exponential kernel), however, a near-parametric rate is attained (see Theorem 10 of \citealp{van2011information}). This occurs because the squared exponential kernel is supported on extremely smooth functions,  making it particularly well suited to approximating the function $r_0(x) = 0$.
  The general problem of lower bounding posterior contraction rates for Gaussian process regression models is studied by \citet{castillo2008lower}.
\end{remark}

As an alternative to establishing consistent model selection, if the goal is inference on $\beta_0$ it can also suffice to instead project the posterior of $(r, \beta)$ onto the family $\sF_0$. Recall that using projections for inference is what we actually recommend. The following two theorems give results for posterior projections of models based on Gaussian process priors; we state both results for the non-fractional posterior.

\begin{theorem}[Semiparametric Bernstein-von Mises Theorem]
  \label{thm:sbvm}
  Consider Model NR with $\sigma_0 = 1$ and $r_0(x) = 0$. Assume $\Sigma_X = \E_{\theta_0}(X_i X_i^\top)$ exists, is finite, and is non-singular. Let $\Pi(\cdot \mid \bZ_N)$ be the posterior associated to the model with $[Y_i \mid X_i, \mu] \sim \Normal\{\mu(X_i), 1\}$ and $\mu \sim \GP(0, \kappa)$ where $\kappa(x,x') = \sigma^2_\beta \, x^\top x' + a \, e^{-\rho \|x - x'\|^2_2}$ for some positive constants $(\sigma^2_\beta, a, \rho)$. Let $\beta^\star = \arg \min_\beta \|\bmu - \bX \beta\|$ where $\bmu = (\mu(X_1), \ldots, \mu(X_N))^\top$ and $\bX = (X_1, \ldots, X_N)^\top$. Then
  \begin{equation*}
    d\left\{ \sqrt N (\beta^\star - \widehat \beta_{LS}),  \Normal(0, \Sigma_X^{-1})\right\} \to 0
    \quad \text{in $F_{\theta_0}$-probability}
  \end{equation*}
  where $d(\cdot, \cdot)$ metrizes convergence in distribution and $\widehat\beta_{LS}$ is the least-squares estimator of $\beta_0$.
\end{theorem}

\begin{remark}
  \citet{xie2021bayesian} study Model~NR in detail, and also allow for $x^\top\beta$ to be replaced by a nonlinear term $g(x; \beta)$ under appropriate conditions. Additionally, they allow $r_0(x) \ne 0$ and impose only that $r_0(x)$ is sufficiently smooth, with $\Sigma_X^{-1}$ replaced with the semiparametric efficient variance in Theorem~\ref{thm:sbvm}. 
\end{remark}

\begin{remark}
  For a Gaussian process prior with Model NR, the posterior distribution of $\beta^\star$ is normal, and so Theorem~\ref{thm:sbvm} really only concerns the behavior of the posterior mean and variance. We note that Theorem~\ref{thm:sbvm} does not directly rule out that a Bernstein-von Mises result holds without the linear term $\sigma^2_\beta x^\top x'$ in the kernel. In Section~\ref{sec:semiparametric-bernstein-von-mises-for-gaussian-processes} we confirm empirically that Theorem~\ref{thm:sbvm} accurately predicts performance, while kernels without the linear term are empirically highly biased in finite samples.
\end{remark}

To show that this phenomenon is general, we prove an analogous result for logistic regression (Model NC). The proof follows Theorem~2.2 of \citet{l2023semiparametric}, which is itself based on Theorem~2.1 of \citet{castillo2015bernstein}. The main steps of the proof are (i) we assume that $\|r - r_0\|_N = \sqrt{N^{-1} \sum_i \mu_0(X_i) \{1 - \mu_0(X_i)\}\{r_0(X_i) - r(X_i)\}^2}$ has a posterior contraction rate faster than $N^{-1/4}$, which allows us to control the remainder of a Taylor expansion and (ii) a change-of-variables argument under the prior $\beta \sim \Normal(\mu_\beta, \Sigma_\beta)$, similar to the application of the Cameron-Martin theorem by \citet{castillo2015bernstein}. On the technical side, guaranteeing that item (i) holds is difficult due to the fact that $\|r - r_0\|_N$ is not equivalent to the empirical Hellinger distance, and so is not available from existing results for Gaussian process logistic regression; as discussed in the Supplementary Material, this difficulty can be bypassed by working with an appropriately-truncated prior instead. As with the results of \citet{xie2021bayesian}, we expect that this result can be modified with additional assumptions to also accommodate $r_0(x)$ deviating from $x^\top \beta_0$ provided that $\|r - r_0\|_N$ has a posterior contraction rate faster than $N^{-1/4}$; after such a modification, the asymptotic variance would be the semiparametric efficient variance.

\begin{theorem}[Bernstein-von Mises for Logistic Regression]
  \label{thm:logistic-bvm}
  Consider Model~NC with $r_0(x) \equiv 0$ and consider a fixed sequence $X_1, X_2, \ldots$ in $[0,1]^P$ and let $\Ihat = \frac{1}{N} \sum_{i = 1}^N \mu_0(X_i) \{1 - \mu_0(X_i)\} X_i \, X_i^\top$. Suppose $\Ihat$ is invertible and that $\Ihat \to I$ for some positive definite matrix $I$. Let $\Pi(\cdot\mid\bZ_N)$ be the posterior distribution of $(r, \beta)$ associated to the model with $[Y_i \mid X_i,r,\beta] \sim \Bernoulli\{(1+e^{- (X_i^\top \beta + r(X_i))})^{-1}\}$ such that Assumption~A in the Supplementary Material is satisfied. Define $\beta^\star = \Ihat^{-1} N^{-1} \sum_i \mu_0(X_i)\{1-\mu_0(X_i)\}X_i \{r(X_i) + X_i^\top\beta\}$ where $\Ihat = N^{-1} \sum_i \mu_0(X_i) \{1 - \mu_0(X_i)\} X_i \, X_i^\top$. Then
  \begin{equation*}
    d\left\{\Pi\left(\sqrt N (\beta^\star - \widehat \beta) \in \cdot \mid \bZ_N\right), \Normal(0, I^{-1})\right\}
    \to 0 \quad \text{in $F_{\theta_0}$-probability}
  \end{equation*}
  where $d(\cdot, \cdot)$ metrizes convergence in distribution and $\widehat \beta = \beta_0 + N^{-1} \Ihat^{-1} \sum_{i = 1}^N \{Y_i - \mu_0(X_i)\} X_i$.
\end{theorem}

Theorem~\ref{thm:logistic-bvm} shows that a certain ``information-weighted'' projection $\beta^\star$ satisfies a semiparametric Bernstein-von Mises theorem with $I$ the efficient Fisher information under $\sF_0$. The quantity $\widehat\beta$ is asymptotically equivalent to the MLE under the parametric logistic regression submodel, being a one-step Newton–Raphson update of a $\sqrt N$-consistent quantity. We use the information-weighted projection mainly for technical convenience, as it makes $\beta^\star$ a linear functional of $\theta(x) = r(x) + x^\top\beta$. In practice, it is much better to use the Kullback-Leibler projection $\beta^\star = \arg \min_{\beta} \sum_i \mu(X_i) \log \frac{\mu(X_i) }{\mu_\beta(X_i)} + \{1 - \mu(X_i)\} \log \frac{1 - \mu(X_i)}{1 - \mu_\beta(X_i)}$ where $\mu_\beta(x) = (1 + e^{-x^\top\beta})^{-1}$, as in experiments we have found it to be far less biased.

\section{Applications}
\label{sec:applications}

\subsection{Rate Adaptivity of General BART}
\label{sec:rate-adaptivity-of-gbart}

We now show empirically that the GBART model satisfies our rate adaptivity result. In particular (i) it performs just as well as a correctly specified linear regression model when the linear model holds, and (ii) it performs as well as BART when the linear model is misspecified.
We consider the data generating mechanism
\begin{equation*}
  Y_i = \frac{1}{\sqrt P} \sum_{j = 1}^P X_{ij} + \lambda_0 \, X_{i1}^2 + \epsilon_i,
  \quad \epsilon_i \sim \Normal(0, \sigma_0^2),
\end{equation*}
with $X_{ij} \sim \iid \Normal(0,1)$. When $\lambda_0=0$ the model is exactly linear, while $\lambda_0 \ne 0$ introduces a quadratic component. We consider $\sigma_0 \in \{1,3,5\}$ and $\lambda_0 \in \{0,0.4\}$. For each $N \in \{2^6,2^7,\ldots,2^{12}\}$ we simulate five datasets and compare predictive accuracy of the following models:
\begin{itemize}
\item \textbf{BART:} the original BART model of \citet{chipman2010bart} with its default prior.  
\item \textbf{GBART:} the general BART model with its default prior and a flat prior on $\beta$.  
\item \textbf{Linear:} a parametric linear regression fit via least squares.  
\end{itemize}

\paragraph{Performance Metrics} Methods are compared by mean squared error (MSE) in estimating $\mu_0(x)=\E_{\theta_0}(Y_i \mid X_i=x)$. The MSE is $\frac{1}{N}\sum_i \{\mu_0(X_i)-\widehat\mu(X_i)\}^2$, where $\widehat\mu(x)$ is the posterior mean of $\mu(x)$. 

\begin{figure}
  \centering
  \includegraphics[width=1\textwidth]{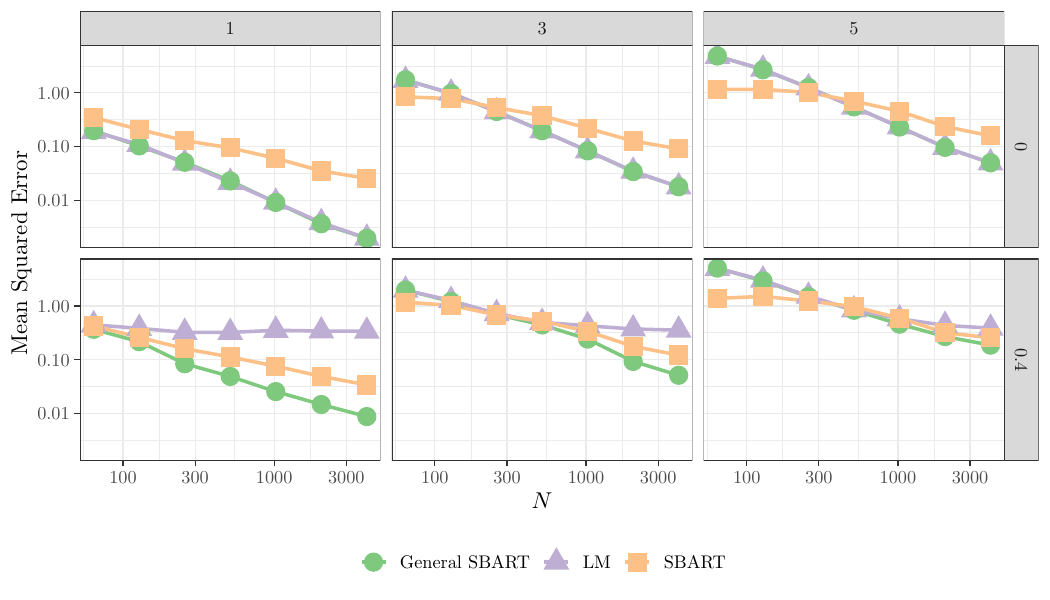}
  \caption{
    Results for the simulation experiment in Section~\ref{sec:rate-adaptivity-of-gbart}. From left to right, we have $\sigma_0 = 1, 3, 5$ and from top to bottom we have $\lambda_0 = 0, 0.4$. MSE and $N$ are displayed on the log-scale.
  }
  \label{fig:results-bart}
\end{figure}

\paragraph{Results} Figure~\ref{fig:results-bart} shows the results. We see that when $\lambda_0 = 0$, the performance of GBART is indistinguishable from that of the linear model at all sample sizes. Somewhat surprisingly, for smaller signal-to-noise ratios and sample sizes the unmodified BART model outperforms a correctly-specified linear model; this is because the sample size is not large enough to estimate the coefficients reliably, giving an advantage to models that regularize the predictions toward zero. This defect in the linear model can be mitigated by replacing the least squares estimates with a ridge regression or lasso.

When $\lambda_0=0.4$, GBART outperforms the other methods, particularly at large sample sizes. Under these simulation settings, it is strictly superior to the linear model under all data generating mechanisms. We again see that when both the sample size and the signal-to-noise ratios are small there is an advantage to using a BART model that more strongly regularizes the predictions.

\subsection{Adaptivity and Selection of a Gaussian Process}

We also confirm empirically that the spike-and-GP prior can consistently select the correct parametric submodel. We use the same simulation design as in Section~\ref{sec:rate-adaptivity-of-gbart}. For the prior on $r(x)$ we take $r \sim \GP(0, \sigma_\mu^2 \kappa)$ with $\sigma_\mu^2 \sim p_0 \, \delta_0 + (1-p_0)\InvGam(a_{\sigma_\mu}, b_{\sigma_\mu})$ and $\kappa(x,x') = \exp\{-\rho \|x-x'\|^2\}$ with $\rho \sim \InvGam(a_\rho, b_\rho)$. In our illustrations we set $a_{\sigma_\mu} = b_{\sigma_\mu} = a_\rho = b_\rho = 1$. The inverse-gamma prior was chosen because it puts little mass near zero, so that the model is genuinely nonparametric when $r(x) \ne 0$. We set $p_0=1/2$ to express prior indifference between parametric and nonparametric models, though in practice one might prefer a smaller $p_0$.

\paragraph{Performance Metrics} We focus on the posterior probability assigned to the nonparametric model, $\Pi(r \ne 0 \mid \bZ_N)$. When $\lambda_0=0$ (linear model), we want $\Pi(r \ne 0 \mid \bZ_N) \to 0$, while when $\lambda_0 \ne 0$ we want $\Pi(r \ne 0 \mid \bZ_N) \to 1$.

\begin{figure}
  \centering
  \includegraphics[width=1\textwidth]{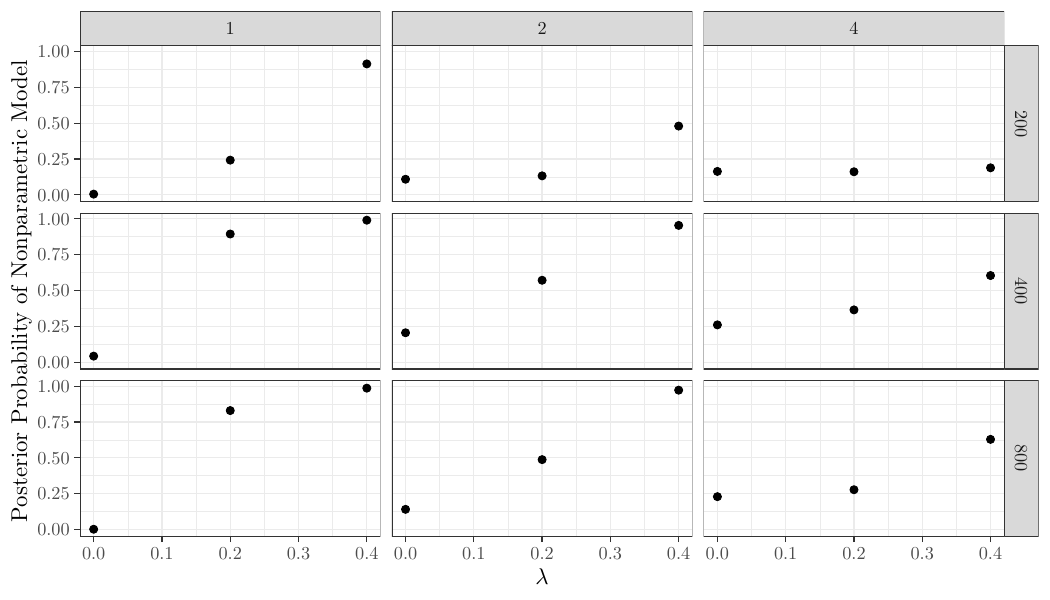}
  \caption{Posterior inclusion probabilities for the Gaussian process model under various values of $\lambda_0$. From left to right: $\sigma_0 \in \{1,2,4\}$. From top to bottom: $N \in \{200,400,800\}$.}
  \label{fig:gp-results}
\end{figure}

\paragraph{Results} Figure~\ref{fig:gp-results} shows that the spike-and-GP prior behaves as expected. When $\lambda_0=0$ the parametric model is consistently selected, while as $\lambda_0$ increases the nonparametric model is selected with high probability. The trend is more muted when $\sigma_0$ is larger; this is intuitive, since (for example) $\sigma_0 = 4$ corresponds to a signal-to-noise ratio of roughly 1-to-16, while $\sigma_0=1$ corresponds to a signal-to-noise ratio of roughly 1-to-1.

As a general lesson, these results suggest that projection-based inference is often preferable to relying solely on model selection consistency. Projections allow us to quantify the loss from using the parametric approximation directly, which can be more informative than a binary model selection outcome. We examine this in more detail in Section~\ref{sec:semiparametric-bernstein-von-mises-for-gaussian-processes}.
\subsection{Semiparametric Bernstein-von Mises for Gaussian Processes}
\label{sec:semiparametric-bernstein-von-mises-for-gaussian-processes}

We now empirically validate the Bernstein-von Mises result for the projection of the GP posterior onto the linear regression model. We use the data generating process
\begin{equation*}
  Y_i = \sum_{j = 1}^P \beta_{0j} X_{ij} + \epsilon_i,
\end{equation*}
with $X_{ij} \sim \Normal(0,1)$ and $\epsilon_i \sim \Normal(0,1)$ independently. We set $P=5$ and $\beta_{0j} \equiv 1.55$. We then fit a GP regression model $Y_i = \mu(X_i) + \epsilon_i$ with $\mu \sim \GP(0,\kappa)$ and compute the posterior distribution of the projection parameter $\beta^\star = (\bX^\top \bX)^{-1}\bX^\top\bmu$. We consider the Laplace kernel $\kappa(x, x') = e^{-\|x - x'\|}$, the squared exponential kernel $\kappa(x,x') = e^{-\|x - x'\|^2}$, and the squared-exponential-plus-linear kernel $\kappa(x,x') = 100 \, (x')^\top x + e^{-\|x - x'\|_2^2}$.

\begin{figure}
  \centering
  \includegraphics[width=1\textwidth]{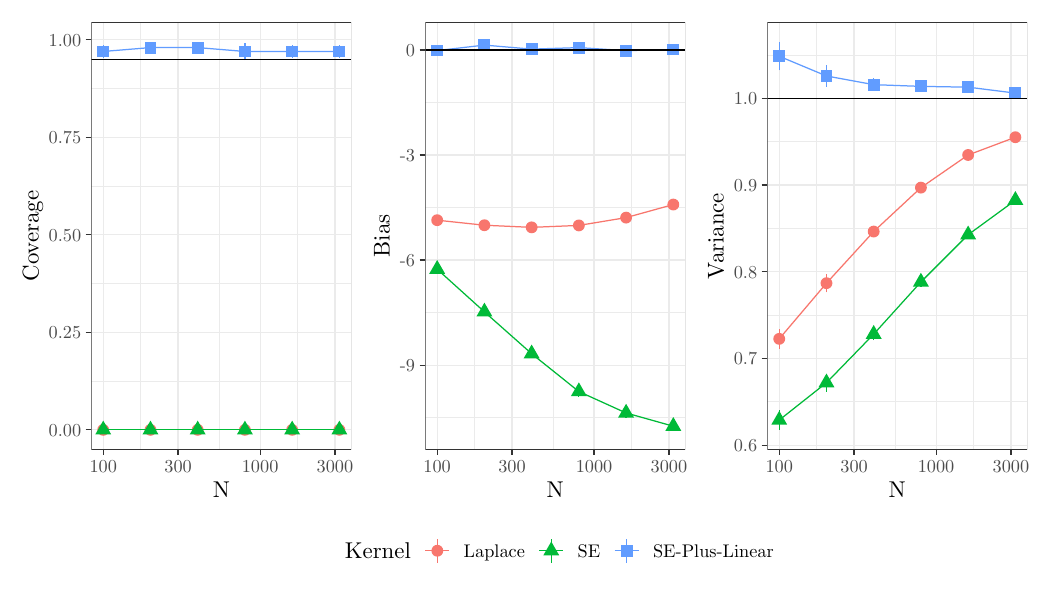}
  \caption{Results for the experiments in Section~\ref{sec:semiparametric-bernstein-von-mises-for-gaussian-processes}.}
  \label{fig:gp-bvm}
\end{figure}

\paragraph{Performance Metrics}
We monitor three quantities: (i) the nominal coverage of 95\% credible intervals for $\beta^\star_1$ under each kernel; (ii) the scaled bias of the posterior mean, defined as the average of $\sqrt N (\widehat\beta_1 - \beta_{01})$ with $\widehat\beta = \E_\Pi(\beta^\star \mid \bZ_N)$; and (iii) the scaled variance, defined as the average of $N \Var_\Pi(\beta^\star_1 \mid \bZ_N)$. If the posterior satisfies a semiparametric Bernstein-von Mises result, we expect coverage near 95\%, scaled bias near zero, and scaled variance near 1, since $\E_{\theta_0}(N^{-1}\bX^\top\bX)=\Identity$ in this setup.

\paragraph{Results}
Figure~\ref{fig:gp-bvm} shows results across kernels and sample sizes. Among the three, only the squared-exponential-plus-linear kernel empirically satisfies the semiparametric Bernstein-von Mises theorem at the chosen sample sizes.
In larger samples, we do see the asymptotic variance of the squared exponential and Laplace kernels approaching the semiparametric efficiency, and indeed asymptotically we do attain the semiparametric efficiency level; the problem with these kernels is that, even in very large samples, the estimates are highly biased, which is clearly seen in the middle panel of Figure~\ref{fig:gp-bvm}.

\subsection{Diagnosing Model Misspecification in Practice}
\label{sec:diagnosing}

We now show how to diagnose model misspecification in practice by analyzing data from the Medical Expenditure Panel Survey (MEPS, \citealp{cohen2009medical}). We consider a subset of the 2011 survey with 16,113 observations, with the goal of predicting self-assessed health status.
We use Model NR with $Y_i$ representing the self-assessed health status of individual $i$ and $X_i$ representing the individual's age, body mass index (BMI), education level, income relative to the poverty line, region of the country, sex, marital status, race, seatbelt use, and smoking status. Although self-assessed health status is ordinal, here we use Model NR as a simple working Gaussian model for prediction and projection-based summarization; extending the same workflow to ordinal or logistic-style models is also possible.

To illustrate the ability of the GBART model to default to a linear model when there is little data while also adapting to nonlinearities and interactions in the data, we consider two analyses: one using a randomly selected sample of size $N = 1000$ and one using a sample of size $N = 8000$. For both cases we monitor the summary $R^2$ of the projection onto a linear regression. All projection summaries reported below are in-sample summaries computed on the analyzed sample. To assess how closely the semiparametric Bernstein-von Mises result holds, we also monitor the posterior distribution of the coefficient of BMI.

\begin{figure}[t]
  \centering
  \includegraphics[width=1\textwidth]{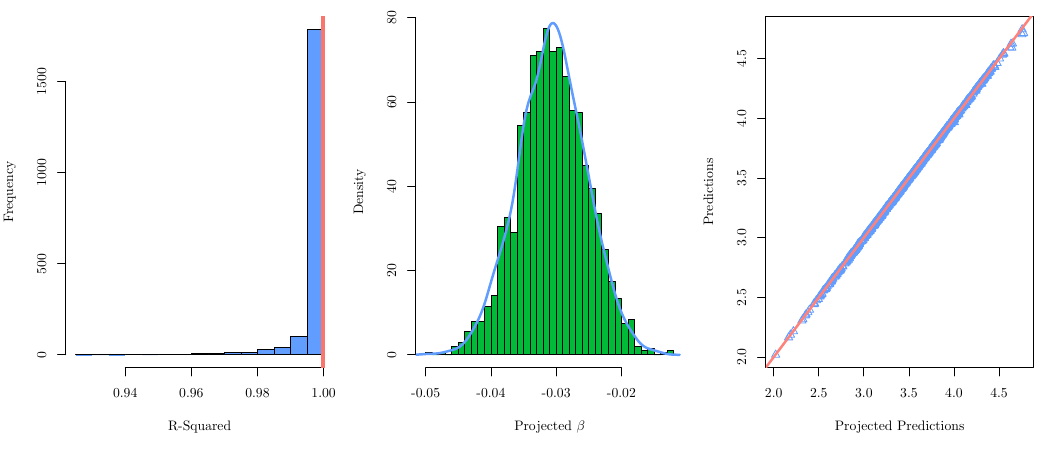}
  \caption{For $N = 1000$: the posterior summary $R^2$ and overall summary $R^2$ of the linear model (left), posterior distribution of the coefficient of BMI compared with the result from the linear model (middle), and plot of the projected predictions against the actual predictions from the model (right).\label{fig:meps1000}}
\end{figure}

Results for the $N = 1000$ model are given in Figure~\ref{fig:meps1000}. We see that the model is shrunk very heavily toward a linear regression model, with the posterior summary $R^2$ being very close to $100\%$. The posterior distribution of the coefficient of BMI is also very close to the posterior from a parametric Bayesian linear model, indicating that the semiparametric Bernstein-von Mises approximation is accurate. Finally, the projected linear regression model gives predictions that are virtually indistinguishable from the full general BART model. This indicates that, at $N = 1000$, there is limited evidence of model misspecification so that the general BART model defaults to the linear model.

\begin{figure}[t]
  \centering
  \includegraphics[width=1\textwidth]{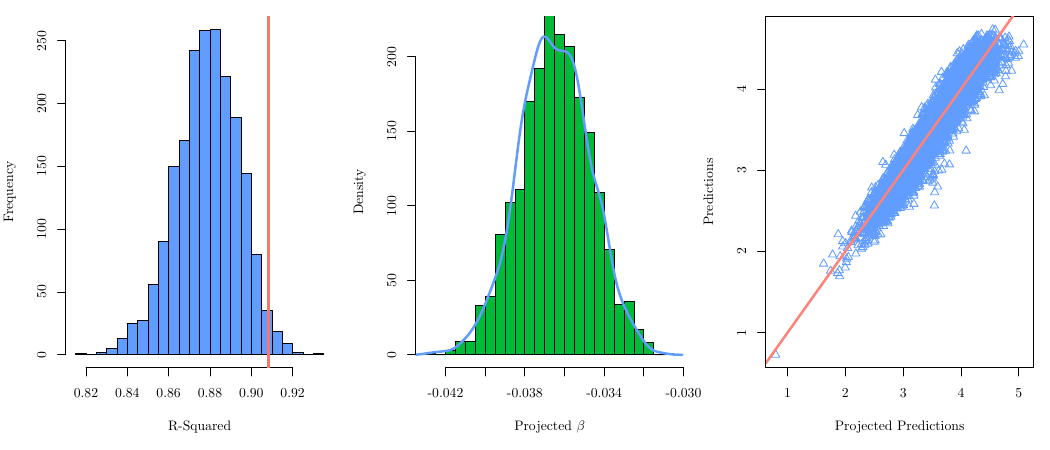}
  \caption{
    For $N = 8000$: the posterior summary $R^2$ and overall summary $R^2$ of the linear model (left), posterior distribution of the coefficient of BMI compared with the result from the linear model (middle), and plot of the projected predictions against the actual predictions from the model (right).\label{fig:meps8000}
  }
\end{figure}

By contrast, results for the $N = 8000$ model are given in Figure~\ref{fig:meps8000}. While the summary $R^2$ of the model is still quite high, we see that the linear regression model no longer captures all of the variability in the predictions. Nevertheless, the semiparametric Bernstein-von Mises result still appears to hold, with the projection parameter $\beta_{\text{BMI}}$ having a posterior distribution in agreement with the parametric linear regression.

\begin{figure}[t]
  \centering
  \includegraphics[width=1\textwidth]{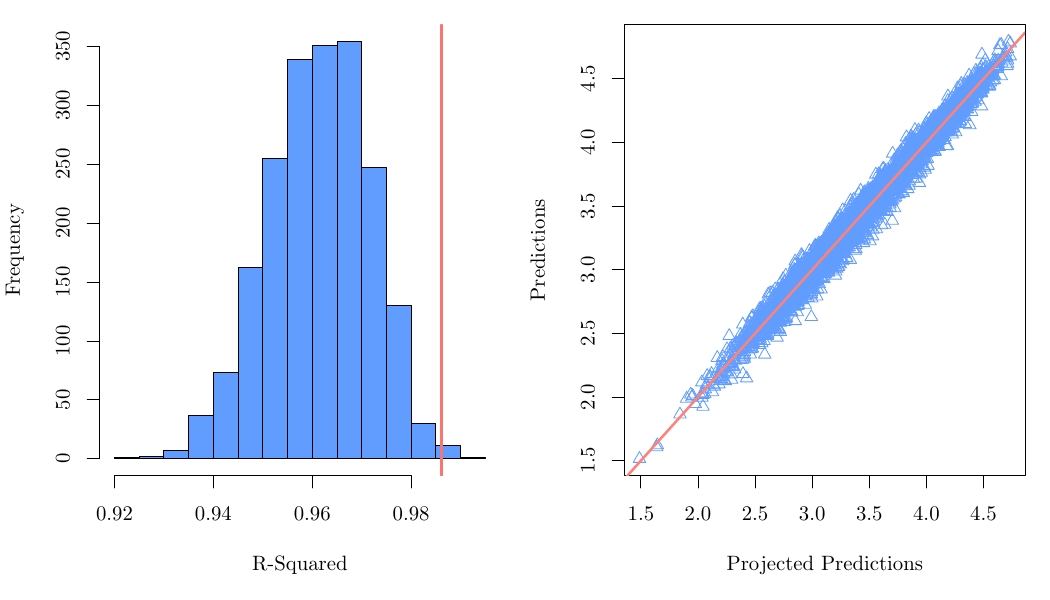}
  \caption{
    For the $N = 8000$ additive model: the posterior summary $R^2$ and overall summary $R^2$ of the linear model (left) and plot of the projected predictions against the actual predictions from the model (right).\label{fig:mepsadditive}
  }
\end{figure}

From the analysis for the $N = 8000$ setting, we see that there is evidence of misspecification of the linear regression, implying that there are either nonlinearities or interactions in the data. An advantage of the fully-Bayesian approach is that we are free to summarize the posterior in any way we choose, and in particular we can look for better models to project the posterior of $\mu(x)$ onto. To this end, in Figure~\ref{fig:mepsadditive} we give the summary $R^2$ and projected predictions for a model that includes spline terms for age, BMI, poverty level, with interaction terms for (age) $\times$ BMI and (age) $\times$ (poverty level). After making this modification, we find that the overall summary $R^2$ of the model exceeds 98\%, with the posterior distribution of $R^2$ centered around 96\%, giving a roughly $10\%$ increase in predictive variability accounted for.

\begin{figure}[t]
  \centering
  \includegraphics[width = 0.7\textwidth]{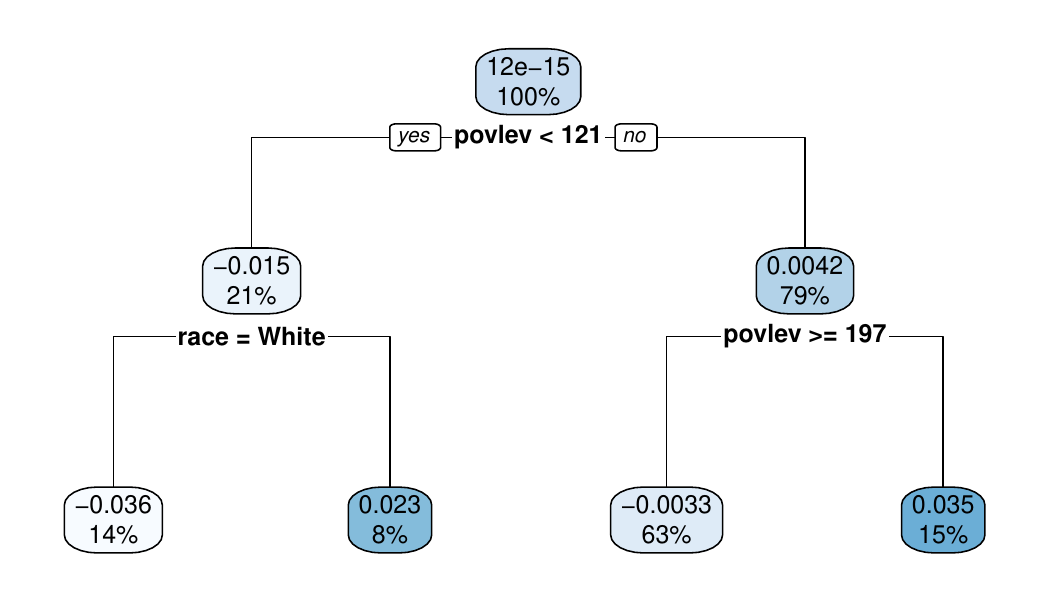}
  \caption{
    Decision tree fit to the residuals $\widehat \mu(X_i) - \widetilde \mu(X_i)$ from the additive model.\label{fig:mepstree}
  }
\end{figure}

To gain better insight into the remaining sources of variability, Figure~\ref{fig:mepstree} displays a CART fit to the residuals $\widehat \mu(X_i) - \widetilde \mu(X_i)$, where $\widetilde \mu(x)$ denotes the additive summary of $\widehat \mu(x)$. The CART predictions quantify the average discrepancy between the additive summary and the GBART fit within each subgroup. The largest degree of misfit occurs for white individuals with poverty level below $121$, with the additive summary predicting a health status on average 0.036 units higher than the GBART model.

\paragraph{Summary} Overall, we see that the general BART model is able to adapt to the amount of evidence in favor of the linear model, reverting to the linear model when there is little evidence of misspecification. In addition, when there \emph{is} evidence of misspecification, we have tools for inspecting the posterior to identify subsets of the data that are responsible for the misspecification. We feel that the additive summary is adequate (0.036 off in a small subgroup, relative to a residual standard deviation of around 1), although we could go further in this example by adding, for example, an interaction between poverty level and race.

\section{Discussion}

The position advocated for in this manuscript is one of defensive model expansion, wherein we embed a parametric Bayesian model inside a nonparametric Bayesian model and apply strong shrinkage toward the parametric model. We claim that this retains many of the benefits of specifying a parametric Bayesian model, such as strong performance when $N$ is small or the signal-to-noise ratio is low, while mitigating the drawbacks of potential model misspecification.

The principal drawback of defensive model expansion is that the expanded nonparametric model may be less interpretable than the embedded parametric model. For example, how should we interpret $\eta$ if the nonparametric component is active in 50\% of posterior samples but inactive in the other 50\%? On this point, an advantage of defensive model expansion is that we also obtain measures of the degree of misspecification. If the summary $R^2$ of the parametric submodel is sufficiently large, we may reasonably conclude that little is lost by interpreting the parametric submodel directly, together with guidelines about ``how wrong'' it is expected to be. We find the Bayesian nonparametric approach particularly appealing for implementing this strategy because it can be carried out coherently through posterior projections, without requiring the analyst to fit alternative models. In this way, defensive model expansion provides both robustness to misspecification and a principled path for interpretation.

\bibliographystyle{apalike}
\bibliography{references.bib}

\end{document}


\maketitle

\tableofcontents

\doublespacing

\section{Proofs}

\subsection{Proof of Theorem~\ref{thm:general-1}}

Let $M_N \uparrow \infty$ and let $\epsilon_N = M_N \, N^{-1/2}$ with $M_N$ such that $\epsilon_N \to 0$ without loss of generality. By PPA we have that for sufficiently large $N$
\begin{equation*}
  \begin{aligned}
    B_{\epsilon_N}
    &\supseteq \{\theta : \|\eta - \eta_0\|^2 + \|r - r_0\|_{\sH}^2 \le  \epsilon_N^2 / C(\theta_0)\}
    \\&\supseteq \{\theta : \max\{\|\eta - \eta_0\|^2,  \|r - r_0\|^2_{\sH}\} \le \epsilon_N^2 / 2C(\theta_0)\}
    = B^\star.
  \end{aligned}
\end{equation*}
For notational convenience, let $L^2_N = M_N^2 / 2C(\theta_0)$. By prior independence, we have for large enough $N$ that
\begin{equation*}
  \begin{aligned}
    \Pi(B^\star)
    &= \Pi_\eta(\|\eta - \eta_0\|^2 \le L_N^2 N^{-1}) \times \Pi_r(\|r - r_0\|^2_{\sH} \le L_N^2 N^{-1})
    \\&\ge
    \frac{\pi(\eta_0) \, C_D}{2}(L_N N^{-1/2})^{P}
    C_1 \, e^{-C_2 \, L_N^2}
    \\&\ge C_1^\star \, e^{-C^\star_2 \, L_N^2},
  \end{aligned}
\end{equation*}
where $C_P > 0$ is the volume of the unit ball in $\mathbb R^P$ and $(C_1^\star, C_2^\star)$ are constants depending on $(\pi(\eta_0), C_P, C_1, C_2)$; the prior-thickness lower bound applies because $L_N \uparrow \infty$, and the polynomial factor can be absorbed into the exponential because $L_N^2 \gg \log N$. The result follows from Theorem 2.1 of \citet{castillo2024bayesian}.

\subsection{Proof of Corollary~\ref{cor:general-1}}

The proof follows immediately from Theorem~\ref{thm:general-1} and the fact that $\Pi(\|r - r_0\|_{\sH} = 0) \ge p_0 > 0$ for all $N$, and so is trivially lower bounded by $e^{-M_N^2}$ for sufficiently large $N$.

\subsection{Proof of Theorem~\ref{thm:general-consistency}}

First, we note that $\Pi^\star_r$ must fail to assign positive mass to $[r = r_0]$, as otherwise the assumption that $\E_{\theta_0} \Pi^\star_\alpha\{H(\theta_0, \theta) < M_N \, N^{-1/2} \mid \bZ_N\} \to 0$ would not hold by Corollary~\ref{cor:general-1}. We write
\begin{equation*}
  \begin{aligned}
    &\Pi_\alpha\{H(\theta_0, \theta) < M_N \, N^{-1/2} \mid \bZ_N\}
    \\&=
    \Pi_{\alpha}\{H(\theta_0, \theta) < M_N \, N^{-1/2} \mid \bZ_N, r = r_0\} \, \Pi_{\alpha}(r = r_0 \mid \bZ_N)
    \\&\quad+ 
    \Pi_{\alpha}\{H(\theta_0, \theta) < M_N \, N^{-1/2} \mid \bZ_N, r \ne r_0\} \, \Pi_{\alpha}(r \ne r_0 \mid \bZ_N)
    \\&=
    \Pi_{\alpha}\{H(\theta_0, \theta) < M_N \, N^{-1/2} \mid \bZ_N, r = r_0\} \, \Pi_{\alpha}(r = r_0 \mid \bZ_N)
    \\&\quad+ 
    \Pi^\star_{\alpha}\{H(\theta_0, \theta) < M_N \, N^{-1/2} \mid \bZ_N\} \, \Pi_{\alpha}(r \ne r_0 \mid \bZ_N),
  \end{aligned}
\end{equation*}
with the last equality holding due to the fact that $\Pi^\star_r$ assigns no mass to $[r = r_0]$. By Corollary~\ref{cor:general-1} we know that this expression must tend to $1$ in $F_{\theta_0}$-probability. However, all of the quantities are bounded by $1$ and by assumption $\Pi^\star_\alpha\{H(\theta_0, \theta) < M_N \, N^{-1/2} \mid \bZ_N\} \to 0$ in $F_{\theta_0}$-probability. Consequently, the only way that $\Pi_{\alpha}\{H(\theta_0, \theta) < M_N \, N^{-1/2} \mid \bZ_N\} \to 1$ in $F_{\theta_0}$-probability can hold is if
\begin{equation*}
  \Pi_{\alpha}\{H(\theta_0, \theta) < M_N \, N^{-1/2} \mid \bZ_N, r = r_0\} \to 1
  \quad \text{and} \quad
  \Pi_{\alpha}(r = r_0 \mid \bZ_N) \to 1
\end{equation*}
in $F_{\theta_0}$-probability, completing the proof.

\subsection{Proof of Theorem~\ref{thm:general-bvm}}

First, we state the additional assumption required to prove this result.

\begin{paragraph}{Assumption BVM}
  The family $\sF_0$ is differentiable in quadratic mean (see \citealp{vandervaart1998asymptotic} Section 5.5), with associated score function $s(\eta)$ and unit Fisher information $\Fisher(\eta)$, such that $\Fisher(\eta_0)$ is non-singular. Additionally, for all $\epsilon > 0$ there exists a sequence of tests $\phi_1, \phi_2, \ldots$ mapping $\sZ^N \to [0,1]$ such that $\E_{(r_0, \eta_0)} \, \phi_N(\bZ_N) \to 0$ and $\sup_{\eta : \|\eta - \eta_0\| > \epsilon} \E_{(r_0, \eta)} 1 - \phi_N(\bZ_N) \to 0$ as $N \to \infty$.
\end{paragraph}

\begin{proof}
  Define $\Pi_\eta^\star(d\eta \mid \bZ_N) = \Pi_\eta(d\eta) \, L(r_0, \eta)^\alpha / \int \Pi_\eta(d\eta) \, L(r_0, \eta)^\alpha$, which has density $\pi^\star(\eta \mid \bZ_N) = \pi_\eta(\eta) \, L(r_0, \eta)^\alpha / \int \pi_\eta(t) \, L(r_0, t)^\alpha \ dt$.

  Let $\widehat \eta$ be any estimator such that $\sqrt N (\widehat \eta - \eta_0) \to \Normal\{0, \Fisher(\eta_0)^{-1}\}$. We note that Assumption BVM, together with PPA, provides the sufficient conditions required for the standard Bernstein-von Mises result for fractional posteriors to hold:
  \begin{equation*}
    \|\Pi^\star_\eta(d\eta \mid \bZ_N) -
    \Normal\{d\eta \mid \widehat \eta, (\alpha N)^{-1} \Fisher(\eta_0)^{-1}\}\|_{TV}
    \to 0 \quad \text{in $F_{\theta_0}$-probability}.
  \end{equation*}
  See, for example, Theorem 10.1 of \citet{vandervaart1998asymptotic} and \citet{l2023semiparametric}. By the triangle inequality, it therefore suffices to show that
  \begin{equation}
    \label{eq:bvm-target}
    \|\Pi_\eta^\star(d\eta \mid \bZ_N) - \Pi_\alpha(d\eta \mid \bZ_N)\|_{TV} \to 0
  \end{equation}
  in $F_{\theta_0}$-probability. Now, consider repeatedly sampling $(r, \eta) \sim \Pi_\alpha(d\theta \mid \bZ_N)$ until $r = r_0$, and define $\eta$ to be the first sample in the sequence and $\eta^\star$ to be the last; note that $\eta^\star \sim \Pi_\eta^\star(d\eta \mid \bZ_N)$ while $\eta \sim \Pi_{\alpha}(d\eta\mid \bZ_N)$. This construction couples the two marginal distributions using the same initial draw whenever that draw already satisfies $r = r_0$. An upper bound on the total variation distance between these two distributions is therefore given by the posterior probability that $\eta^\star \ne \eta$, which is upper bounded by the posterior probability that $r \ne r_0$ (see \citealp{levinMarkovChainsMixing2009} Proposition 4.7, for example). Hence
  \begin{equation*}
    \|\Pi_\eta^\star(d\eta \mid \bZ_N) - \Pi_\alpha(d\eta \mid \bZ_N)\|_{TV} \le \Pi_\alpha(r \ne r_0 \mid \bZ_N).
  \end{equation*}
  (Note: this argument implicitly assumes that $\Pi_\alpha(r \ne r_0 \mid \bZ_N) < 1$, but the bound still holds trivially if $\Pi_\alpha(r \ne r_0 \mid \bZ_N) = 1$.)
  But by Theorem~\ref{thm:general-consistency} we know $\Pi_\alpha(r \ne r_0 \mid \bZ_N) \to 0$ in $F_{\theta_0}$-probability, which implies \eqref{eq:bvm-target}.
\end{proof}

\subsection{Proof of Theorem~\ref{thm:bart-adaptivity}}

First, PPA holds for this setting with $\|r - r_0\|_{\sH}$ the supremum norm $\sup_{x \in [0,1]^D} |r(x) - r_0(x)|$. Additionally, with positive probability all of the decision trees in the ensemble will contain no splits, which reduces to the model
\begin{equation*}
  Y_i = r + X_i^\top\beta + \epsilon_i
\end{equation*}
where $r \sim \Normal(0, \sigma^2_\mu)$ independently of $\beta$. This is a parametric linear regression model with intercept $r$ (or $r + \beta_1$ if $X_{i1}$ is an intercept term), and therefore by standard arguments assigns sufficient mass to neighborhoods of $(\beta_0, \sigma_0)$ to guarantee the rate.

The fact that $\Pi_\alpha(\sN \mid \bZ_N) \to 1$ in $F_{\theta_0}$-probability is a consequence of Proposition~4.3 of \citet{tan2024computational}.

\subsection{Proof of Theorem~\ref{thm:bart-bvm}}

By Theorem~\ref{thm:bart-adaptivity}, $\Pi_{\alpha}(\sN \mid \bZ_N) \to 1$ in $F_{\theta_0}$-probability, and hence the model reduces with probability tending to 1 to
\begin{equation*}
  Y_i = (\beta^\star_1, \beta_2, \ldots, \beta_D)^\top X_i + \epsilon_i,
  \quad \epsilon_i \sim \Normal(0, \sigma^2),
\end{equation*}
where $\beta^\star_1 = r(0) + \beta_1$. This reduced model is differentiable in quadratic mean and satisfies the conditions of Theorem 10.1 of \citet{vandervaart1998asymptotic} (the additional condition that $\Sigma_X$ is non-singular is needed for the testing condition to hold). The rest of the proof follows as in the proof of Theorem~\ref{thm:general-bvm} with the event $[r = r_0]$ replaced with the event $\sN$.

\subsection{Proof of Theorem~\ref{thm:sbvm}}

Following routine calculations, one can show that $\beta^\star = (\bX^\top \bX)^{-1} \bX^\top \bmu$. This is a linear functional of $\mu(\cdot)$ and hence, because the posterior of $\mu(\cdot)$ is a Gaussian process, $\beta^\star$ is normally distributed.

Let $K \in \Reals^{N \times N}$ have $(i, i')^{\text{th}}$ entry given by $\kappa(X_i, X_{i'})$. Then the posterior distribution of $\bmu$ is given by
\begin{equation*}
  \bmu \sim \Normal\{K \, (K + \Identity)^{-1} \bY, K - K (K + \Identity)^{-1} K\}
\end{equation*}
using properties of the multivariate Gaussian distribution
\begin{equation*}
  \binom{\bY}{\bmu} \sim \Normal\left\{ \binom{0}{0}, \begin{pmatrix} K + \Identity & K \\ K & K \end{pmatrix} \right\}.
\end{equation*}
Consequently, the posterior of $\beta^\star$ is
\begin{equation*}
  \beta^\star \sim \Normal\{
  B K (K + \Identity)^{-1} \bY,
  B [K - K(K + \Identity)^{-1} K] B^\top
  \}
\end{equation*}
where $B = (\bX^\top \bX)^{-1} \bX^\top$. Noting that $\widehat \beta = B \, \bY$, our goal is to find the limiting values of
\begin{equation*}
  \sqrt N \, B \left\{ K (K + \Identity)^{-1} - \Identity \right\} \bY
  \quad \text{and} \quad
  N B [K - K (K + \Identity)^{-1} K] B^\top.
\end{equation*}
We begin by noting the identity $K - K (K + \Identity)^{-1} K = K (K + \Identity)^{-1} = \Identity - (K + \Identity)^{-1}$, which can be proved by eigendecomposition. From this, the variance expression reduces to
\begin{equation*}
  N B K \, (K + \Identity)^{-1} B^\top = \widehat \Sigma_N^{-1} \{N^{-1} \bX^\top [\Identity - (K + \Identity)^{-1}] \bX\} \widehat \Sigma_N^{-1} \quad \text{where} \quad \widehat \Sigma_N = N^{-1} \sum_i X_i \, X_i^\top.
\end{equation*}
Because $\widehat \Sigma_N \to \Sigma_X$ it suffices to show that $\mathcal R_N = \frac{1}{N} \bX^\top (K + \Identity)^{-1} \bX \to \zeros$. This is immediate from the to-be-shown fact that $\bX^\top (K + \Identity)^{-1} \bX \preceq \bX^\top \bX / (1 + c \, N)$ for some positive constant $c$ with probability tending to $1$, which uses the linear component of $K$, but the following argument shows that this holds even without the linear component. To do this we show that $\|\mathcal R_N\|_{\text{op}}$ given by the largest eigenvalue of $\mathcal R_N$ also tends to zero.
To this end, let $\Gamma \, D \, \Gamma^\top$ denote the eigendecomposition of $(\Identity + K)^{-1}$, which has diagonal entries $1 / (1 + \lambda_r)$ where $\lambda_N \le \lambda_{N-1} \le \cdots \le \lambda_1$ are the eigenvalues of $K$. Then for any $\beta \in \Reals^P$ with $\|\beta\| = 1$ we have
\begin{equation*}
  \beta^\top \mathcal R_N \beta = \frac{1}{N} \sum_{r = 1}^N \frac{(\gamma_r^\top \, \bX \, \beta)^2}{1 + \lambda_r},
\end{equation*}
where $\gamma_r$ is the $r^{\text{th}}$ eigenvector of $(\Identity + K)^{-1}$. Fix an integer $M \ge 1$ to be chosen later. We can then split the sum into two pieces
\begin{equation*}
  \beta^\top \mathcal R_N \beta \le
  \underbrace{\frac{1}{N}\,\frac{1}{1+\lambda_M}\,\|\Gamma_M^\top \bX\beta\|_2^2}_{\text{top }M}
  +
  \underbrace{\frac1N \|(\Identity-\Gamma_M\Gamma_M^\top)\bX\beta\|_2^2}_{\text{tail}},
\end{equation*}
where $\Gamma_M = [\gamma_1, \ldots, \gamma_M]$ contains the first $M$ eigenvectors, and the inequality follows from setting $(1 + \lambda_r) \ge 1$ for $r > M$ and $(1 + \lambda_r) \ge 1 + \lambda_M$ for $r \le M$. Therefore,
\begin{equation*}
  \|\mathcal R_N\|_{\mathrm{op}}
  = \sup_{\|\beta\|=1}\beta^\top \mathcal R_N \beta
  \le
  {\frac{1}{1+\lambda_M}\times \frac{\|\bX\|_{\mathrm{op}}^2}{N}}
  +
  {\sup_{\|\beta\|=1}\frac 1 N \|(\Identity-P_M)\bX\beta\|_2^2},
  \quad \text{where} \quad P_M= \Gamma_M \Gamma_M^\top .
\end{equation*}
Let $\mu_M>0$ denote the $M$th eigenvalue of the Mercer operator associated with $\kappa$. For fixed $M$, the $M$th eigenvalue $\lambda_M$ of the unnormalized empirical Gram matrix satisfies $\lambda_M/N \to \mu_M$ in probability, and hence $\lambda_M \asymp N$. In particular, if $\kappa$ is universal on a compact domain, then $\mu_M>0$ for every fixed $M$.
For the second term, let $\kappa(x, x') = \sum_{m = 1}^\infty \mu_m \, \psi_m(x) \, \psi_m(x')$ denote the spectral decomposition of $\kappa(x,x')$. Because $\kappa(x,x')$ is a universal kernel, each linear coordinate $f_j(x) = x_j$ lies in the $L^2(F_X)$-closure of $\text{span}(\psi_1, \psi_2, \ldots)$; therefore, for any $\epsilon > 0$ there exists $M$ such that
\begin{align*}
  \sup_{\|\beta\| = 1} \|f_\beta - H_M f_\beta\|^2_{L^2(F_X)} \le \epsilon,
  \quad \text{where} \quad f_\beta(x) = x^\top \beta
\end{align*}
and $H_M$ is the projection operator onto the span of $\{\psi_1, \ldots, \psi_M\}$. By the uniform law of large numbers and convergence of the empirical projection operators (see \citealp{koltchinskii2000random}) we have
\begin{equation*}
  \sup_{\|\beta\|=1}\frac 1 N \|(\Identity-P_M)\bX\beta\|_2^2
  \to \sup_{\|\beta\| = 1}\|f_\beta - H_M \, f_\beta\|_{L_2(F_X)}^2 \le \epsilon \quad \text{$F_{\theta_0}$-almost surely.}
\end{equation*}
Hence,
\begin{equation*}
  \limsup \|\mathcal R_N\|_{\text{op}}
  \le 
  \epsilon \qquad \text{$F_{\theta_0}$-almost surely},
\end{equation*}
which completes the proof for the variance term because $\epsilon > 0$ was arbitrary.

Next, we address the asymptotic bias. From the previously noted identity $K (K + \Identity)^{-1} = \Identity - (K + \Identity)^{-1}$ our goal will be to find the limiting value of
\begin{equation*}
  \sqrt N \, B \, (K + \Identity)^{-1} \bY.
\end{equation*}
To prove the bias condition, we proceed in two steps:
\begin{enumerate}
\item We show that the result holds whenever $K$ satisfies $\bX^\top (K + \Identity)^{-1} \bX \preceq \frac{\bX^\top \bX}{1 + c \, N}$ for some $c > 0$ with probability tending to $1$, where $A \preceq B$ for positive semi-definite matrices means that $B - A$ is a positive semi-definite matrix.
\item We show that our choice of $K$ satisfies this condition.
\end{enumerate}
We first decompose into bias and noise terms as
\begin{equation*}
  \sqrt N (\bX^\top \bX)^{-1} \bX^\top (K + \Identity)^{-1} \, \bY
  = \underbrace{\frac{1}{\sqrt N}\widehat \Sigma_N^{-1}\bX^\top (K + \Identity)^{-1} \bX\,\beta_0}_{\text{Bias}}
  + \underbrace{\frac{1}{\sqrt N}\widehat \Sigma_N^{-1}\bX^\top (K + \Identity)^{-1} \epsilon}_{\text{Noise}}.
\end{equation*}
where $\epsilon \sim \Normal(0, \Identity)$. To handle the bias term, we show that its squared norm tends to $0$. This is given by
\begin{equation*}
  \frac{1}{N} \beta_0^\top \bX^\top (K + \Identity)^{-1} \bX \widehat \Sigma_N^{-1} \widehat \Sigma_{N}^{-1} \bX^\top (K + \Identity)^{-1} \bX \, \beta_0.
\end{equation*}
Applying our assumption, we can bound the expression above (on a set with probability approaching 1) as
\begin{equation*}
  \frac{1}{N (1 + cN)^2} \beta_0^\top \bX^\top \bX \widehat \Sigma_N^{-2} \bX^\top \bX \beta_0
  = \frac{N^2}{N (1 + cN)^2} \beta_0^\top \beta_0 \stackrel{N \to \infty}{\longrightarrow} 0.
\end{equation*}
Similarly, we show that the expected squared norm of the noise term also tends to zero. As the noise term is mean-$0$, our goal is to establish that
\begin{equation*}
  \Var(\text{Noise})
  =
  \frac{1}{N} \widehat \Sigma_N^{-1} \bX^\top (K + \Identity)^{-1} \bX \widehat \Sigma_N^{-1} \stackrel{N \to \infty}{\longrightarrow} 0.
\end{equation*}
As above, this follows from the fact that $\bX^\top (K + \Identity)^{-1} \bX \preceq \frac{\bX^\top\bX}{1 + c \, N}$.

It remains to establish that $\bX^\top (K + \Identity)^{-1} \bX \preceq \frac{\bX^\top \bX}{1 + c \, N}$ for some constant $c$. To do this, we apply the following fact: for any positive definite matrix $M$ and subspace $\mathcal S$, if it holds for some $m > 0$ that
\begin{math}
  v^\top M v \ge m \, \|v\|^2
\end{math}
for all $v \in \mathcal S$ then it also holds that $v^\top M^{-1} v \le \frac{1}{m} \|v\|^2$ for all $v \in \mathcal S$. In this case, take
$v$ to be any member of the column space of $\bX$ and let $M = K + \Identity$. Then $v^\top M v = v^\top (\sigma^2_\beta \bX \, \bX^\top + K^\star + \Identity) v \ge \|v\|^2 + \sigma^2_\beta v^\top \bX \, \bX^\top v$ where $K^\star$ corresponds to the piece of the kernel matrix for the squared exponential term. But for $v$ in the column space of $\bX$ we have
\begin{equation*}
  \|v\|^2 + \sigma^2_\beta v^\top \bX \bX^\top v
  \ge \|v\|^2 (1 + c^\star \, N) \quad \text{where} \quad c^\star = N^{-1} \sigma^2_\beta \, \lambda_{\text{min}}(\bX^\top\bX) = \sigma^2_\beta \, \lambda_{\text{min}}(\widehat \Sigma_N)
\end{equation*}
and where $\lambda_{\text{min}}(A)$ denotes the minimal eigenvalue of $A$. As $\lambda_{\text{min}}(\widehat \Sigma_N) \cinp \lambda_{\text{min}}(\Sigma_X)$ we can take $c = \sigma^2_\beta \, \lambda_{\min}(\Sigma_X) / 2$.

\subsection{Proof of Theorem~\ref{thm:logistic-bvm}}

We begin by noting that it suffices to show that, for each fixed $\lambda \in \Reals^P$, the Bernstein-von Mises result holds for $\sqrt N (\lambda^\top \beta^\star - \lambda^\top \widehat \beta)$, that is, that the posterior converges to a $\Normal(0, \lambda^\top I^{-1} \lambda)$ distribution. This allows us to apply strategies for posterior convergence of scalar functionals as studied by \citet{castillo2015bernstein} and \citet{l2023semiparametric}.

Next, we introduce some additional notation. Let $\theta(x) = x^\top\beta + r(x)$ and define the log-likelihood as $\ell_N(\theta) = \sum_i f_i(\theta)$ where $f_i(\theta) = Y_i \, \theta(X_i) - \log(1 + e^{\theta(X_i)})$ is the unit-level log-likelihood. Further, adopting the shorthand $\mu_\theta(x) = (1 + e^{-\theta(x)})^{-1}$, let $f'_i(\theta) = \{Y_i - \mu_\theta(X_i)\}$ and $f''_i(\theta) = - \mu_\theta(X_i) \, \{1 - \mu_\theta(X_i)\}$ denote the first and second derivatives of $f_i(\theta)$ with respect to $\theta(X_i)$. Let $w_i = -f''_i(\theta_0)$. Define the observed unit Fisher information $\Ihat = N^{-1} \sum_i w_i \, X_i \, X_i^\top$, the empirical process $W_N(h) = N^{-1/2} \sum_{i=1}^N \{Y_i - \mu_0(X_i)\} \, h(X_i)$, the inner product $\langle a, b \rangle_N = N^{-1} \sum_i w_i \, a(X_i) \, b(X_i)$ with associated norm $\|\cdot\|_N$, and the influence function $\psi_0(x) = \lambda^\top \Ihat^{-1} \, x$. We then observe that $\lambda^\top \beta^\star - \lambda^\top \widehat \beta = \langle \theta - \theta_0, \psi_0 \rangle_N - W_N(\psi_0) / \sqrt N$. Finally, we define a local perturbation of $\theta$ as $\theta_t(x) = \theta(x) - h_t(x)$ where $h_t(x) = t \, \psi_0(x) / \sqrt N$.

We impose the following further assumptions referenced in the statement of Theorem~\ref{thm:logistic-bvm}.

\paragraph{Assumption A}
Suppose that the conditions of Theorem~\ref{thm:logistic-bvm} hold. Then we assume further the following conditions on $\Pi(dr,d\beta)$ and $(r_0, \beta_0)$.
\begin{itemize}
\item \textbf{Prior thickness:} Let
  \begin{align*}
    B_{\epsilon_N} = \left\{ (r,\beta) : \E_{r_0, \beta_0} [\ell_N(\theta_0) - \ell_N(\theta)]
    \le N \, \epsilon_N^2
    \quad \text{and} \quad 
    \Var_{r_0, \beta_0} [\ell_N(\theta_0) - \ell_N(\theta)]
    \le N \, \epsilon_N^2
    \right\}.
  \end{align*}
  Then $\Pi(B_{\epsilon_N}) \ge C_1 e^{-C_2 \, N \, \epsilon_N^2}$ for $\epsilon_N = N^{-1/4 - \delta}$ for some small $0 < \delta < 1/4$.
\item \textbf{$L_2$ posterior convergence:} For every constant $K$, we have $\Pi(\|\theta - \theta_0\|_N > K \, N^{-1/4 - \delta} \mid \bZ_N) \cinp 0$ where $\|\theta - \theta_0\|^2_N = \frac{1}{N} \sum_i \mu_0(X_i) \{1 - \mu_0(X_i)\} \, \{\theta(X_i) - \theta_0(X_i)\}^2$. 
\item \textbf{Prior for $\beta$:} The prior is $\Pi(dr, d\beta) = \Pi(dr) \ \Pi(d\beta)$ where $\Pi(d\beta)$ is a normal distribution with mean $\mu_\beta$ and covariance $\Sigma_\beta$.
\end{itemize}

We discuss the $L_2$-posterior convergence rate assumption after the proof. The prior thickness assumption is usually also invoked to prove $L_2$-posterior convergence, but we state it explicitly as it also implies the evidence lower bound
\begin{align*}
  \int \exp\left\{ \ell_N(\theta) - \ell_N(\theta_0) \right\} \ \Pi(dr) \ \Pi(d\beta)
  \ge C_1 e^{-(C_2 + 2) \ N \, \epsilon_N^2}
\end{align*}
holds with probability at least $1 - (N \epsilon_N^2)^{-1}$ (see Lemma~1.1 of \citealp{castillo2024bayesian} with $C = 1$).

We now give a nonparametric local asymptotic normality (LAN) result, which is a starting point for most Bernstein-von Mises type results.

\begin{lemma}
  \label{lem:lem1}
  Under the definitions above, we can write
  \begin{align*}
    \ell_N(\theta) - \ell_N(\theta_t) &= t \, W_N(\psi_0) - \frac{N\{\|\theta - \theta_0\|^2_N - \|\theta_t - \theta_0\|^2_N\}}{2}
    + R(\theta, \theta_t)
  \end{align*}
  where
  \begin{align*}
    |R(\theta, \theta_t)| \le C \, t \|\Ihat^{-1} \, \lambda\| \sqrt N [\|\theta - \theta_0\|^2_N + \|\theta_t - \theta_0\|^2_N +
      \|\theta - \theta_0\|_N \times \|\theta_t - \theta_0\|_N],
  \end{align*}
  for some constant $C$ depending only on $\theta_0$. In particular, if $A_N$ is contained in $\{(r,\beta) : \|\theta - \theta_0\|_N \le \epsilon_N\}$ with $N^{1/4} \, \epsilon_N \to 0$ then $\sup_{\theta \in A_N} |R(\theta, \theta_t)| = o(1)$. Additionally, $t \sqrt N (\lambda^\top \beta^\star - \lambda^\top \widehat \beta) + \ell_N(\theta) - \ell_N(\theta_t) - R(\theta, \theta_t) = \frac{t^2 \lambda^\top \Ihat^{-1} \lambda }{2}$.
\end{lemma}

\begin{proof}[Proof of lemma]
  First,
  \begin{align*}
    f_i(\theta) - f_i(\theta_t) = \int_{\theta_t(X_i)}^{\theta(X_i)} f'_i(u) \ du
  \end{align*}
  where, abusing notation, we take (for example) $f_i'(u) = \{Y_i - (1 + e^{-u})^{-1}\}$. We then expand $f'_i(u)$ about $\theta_0(X_i)$ to obtain
  \begin{align*}
    f_i(\theta) - f_i(\theta_t) &= h_t(X_i) f'_i(\theta_0) + \frac{[\{\theta(X_i) - \theta_0(X_i)\}^2 - \{\theta_t(X_i) - \theta_0(X_i)\}^2]}{2} f''_i(\theta_0)
      \\&\quad+ \int_{\theta_t(X_i)}^{\theta(X_i)} \frac{f'''_i(u^\star) (u - \theta_0(X_i))^2}{2} \ du,
  \end{align*}
  for some $u^\star$ between $u$ and $\theta_0(X_i)$. Summing over $i$, we obtain
  \begin{align*}
    \ell_N(\theta) - \ell_N(\theta_t) &= t \, W_N(\psi_0) - \frac{N\{\|\theta - \theta_0\|^2_N - \|\theta_t - \theta_0\|^2_N\}}{2}
    + R(\theta, \theta_t)
  \end{align*}
  where, by the triangle inequality and the fact that $|f'''_i(u^\star)|$ is uniformly bounded, we have
  \begin{align*}
    |R(\theta, \theta_t)| &\le
      C_f \, \sum_i |\{\theta(X_i) - \theta_0(X_i)\}^3 - \{\theta_t(X_i) - \theta_0(X_i)\}^3|,
    \\&= C_f \, \sum_i |h_t(X_i)| [\{\theta(X_i) - \theta_0(X_i)\}^2  \\&\qquad +\{\theta(X_i) - \theta_0(X_i)\}\{\theta_t(X_i) - \theta_0(X_i)\} + \{\theta_t(X_i) - \theta_0(X_i)\}^2]
  \end{align*}
  for some $C_f > 0$. Now, by Cauchy-Schwarz we have $|h_t(X_i)| \le t \|\Ihat^{-1} \lambda\| \times  \|X_i\| / \sqrt N = O(N^{-1/2})$ because $\|X_i\|$ is bounded and $\Ihat \to I$. Consequently, by another application of Cauchy-Schwarz we have 
  \begin{align*}
    |R(\theta, \theta_t)| = O(N^{1/2}) \{\|\theta - \theta_0\|_N^2 + \|\theta_t - \theta_0\|_N^2 + \|\theta - \theta_0\|_N \, \|\theta_t - \theta_0\|_N\}.
  \end{align*}
  By the definition of $A_N$ in our lemma (and the fact that $\|\theta - \theta_0\|_N = o(N^{-1/4})$ implies $\|\theta_t - \theta_0\|_N = o(N^{-1/4})$), each of the terms in the above summation is uniformly $o(N^{-1/2})$ so that $\sup_{\theta \in A_N} |R(\theta, \theta_t)| = o(1)$.
  The statement that $t \sqrt N(\lambda^\top \beta^\star - \lambda^\top \widehat \beta) + \ell_N(\theta) - \ell_N(\theta_t) - R(\theta, \theta_t) = \frac{t^2 \, \lambda^\top \Ihat^{-1} \lambda}{2}$ follows from applying the definitions of the various terms and simplifying.
\end{proof}

We use the Gaussian prior for $\beta$ both for its algebraic properties under changes of variables and, as in the lemma below, to ensure that the posterior assigns negligible probability to $\beta$ being too large.

\begin{lemma}
  \label{lem:2}
  Under Assumption~A, we have $\Pi(\|\beta - \mu_\beta\| > N^{1/4} \mid \bZ_N) \cinp 0$.
\end{lemma}

\begin{proof}
  From Lemma~4 of \citet{rhee1986uniform}, the prior satisfies $\Pi(\|\beta - \mu_\beta\| > N^{1/4}) \le \exp\left( -C N^{1/2} \right)$ for some constant $C$ depending on $(\mu_\beta, \Sigma_\beta)$. From Lemma~1.2 of \citet{castillo2024bayesian}, the fact that prior thickness holds with $N \epsilon_N^2 = N^{1/2 - 2\delta} \ll N^{1/2}$ implies that $\Pi(\|\beta - \mu_\beta\| > N^{1/4} \mid \bZ_N) \cinp 0$.
\end{proof}

We use the same argument as the proof of Theorem~5.2 of \citet{castillo2024bayesian} which, in view of the fact that we have verified the nonparametric LAN condition, reduces to checking that
\begin{align*}
  \frac{\int_{A_N} e^{\ell_N(\theta_t)} \ \Pi(dr) \ \Pi(d\beta)}{\int e^{\ell_N(\theta)} \ \Pi(dr) \ \Pi(d\beta)} = 1 + o_P(1)
\end{align*}
holds for some set $A_N$ with $\Pi(A_N \mid \bZ_N) = 1 + o_P(1)$. We adopt the shorthand $b = \frac{t \Ihat^{-1} \lambda}{\sqrt N}$. We set $A_N = B_N \cap C_N$ where $B_N = \{(r,\beta) : \|\theta - \theta_0\|_N \le N^{-1/4 - \delta}\}$ and $C_N = \{\|\beta - \mu_\beta\| \le N^{1/4}\}$. From Assumption~A and Lemma~\ref{lem:2}, we have $\Pi(A_N \mid \bZ_N) \cinp 1$. As argued by \citet{castillo2015bernstein}, letting $\Delta = \sqrt N (\lambda^\top \beta^\star - \lambda^\top \widehat \beta)$, it is sufficient to show that
\begin{align*}
  \E(e^{t \Delta} \mid \bZ_N, A_N) = (1 + o_P(1)) e^{t^2 \lambda^\top I^{-1} \lambda / 2}
\end{align*}
holds for all $t \in \Reals$. 
To show that, we compute the left-hand-side as
\begin{align*}
  \frac{\int_{A_N} e^{t \Delta + \ell_N(\theta) - \ell_N(\theta_t)} \, e^{\ell_N(\theta_t)} \ \Pi(dr) \ \Pi(d\beta)}{\int_{A_N} e^{\ell_N(\theta)} \ \Pi(dr) \ \Pi(d\beta)}
  = e^{t^2 \lambda^\top \Ihat^{-1} \lambda / 2} \frac{\int_{A_N} e^{R(\theta, \theta_t)} e^{\ell_N(\theta_t)}\ \Pi(dr) \ \Pi(d\beta)}{\int_{A_N} e^{\ell_N(\theta)} \ \Pi(dr) \ \Pi(d\beta)}.
\end{align*}
By Lemma~\ref{lem:lem1}, the fact that $\Ihat \to I$, and the fact that $\Pi(A_N \mid \bZ_N) = 1 + o_P(1)$, this is
\begin{align}
  \label{eq:finally}
  (1 + o_P(1)) e^{t^2 \lambda^\top I^{-1} \lambda / 2} \frac{\int_{A_N} e^{\ell_N(\theta_t)} \ \Pi(dr) \ \Pi(d\beta)}{\int e^{\ell_N(\theta)} \ \Pi(dr) \ \Pi(d\beta)}.
\end{align}
It therefore suffices to show that $\int_{A_N} e^{\ell_N(\theta_t)} \ \Pi(dr) \ \Pi(d\beta) = (1 + o_P(1)) \int e^{\ell_N(\theta)} \ \Pi(dr) \ \Pi(d\beta)$. To do this, we make the change of variables $\beta \mapsto \beta - b$, which gives
\begin{align*}
  \int_{A_{Nt}} \exp\left\{ \ell_N(\theta) - (\beta - \mu_\beta)^\top \Sigma_\beta^{-1} b - \frac{b^\top \Sigma_\beta^{-1} b}{2} \right\}
  \ \Pi(dr) \ \Pi(d\beta),
\end{align*}
where $A_{Nt} = \{\|\theta_{-t} - \theta_0\|_N \le N^{-1/4 - \delta}\} \cap \{\|\beta + b - \mu_\beta\| \le N^{1/4}\}$. Next, we remove the auxiliary terms from the exponential. From the definition of $b$ and the fact that $\Ihat^{-1} \to I^{-1}$, we have $b^\top \Sigma^{-1}_\beta b / 2 = O(N^{-1})$. For the inner product, note that on $A_{Nt}$ we have $\|\beta - \mu_\beta\| \le \|\beta - \mu_\beta + b\| + \|b\| \le N^{1/4} + \|b\|$; applying Cauchy-Schwarz we have uniformly on $A_{Nt}$ that
\begin{align*}
  |(\beta - \mu_\beta)^\top \Sigma_\beta^{-1} b| \le \|\beta - \mu_\beta\| \times \|\Sigma_\beta^{-1} b\|
  \le \{N^{1/4} + \|b\|\} \|\Sigma_\beta^{-1} b\| = O(N^{-1/4}).
\end{align*}
Consequently, we can rewrite the integral as
\begin{math}
  (1 + o_P(1)) \int_{A_{Nt}} e^{\ell_N(\theta)} \ \Pi(dr) \ \Pi(d\beta),
\end{math}
so it suffices to show that
\begin{align*}
  \int_{A_{Nt}} e^{\ell_N(\theta)} \ \Pi(dr) \ \Pi(d\beta) = (1 + o_P(1)) \int e^{\ell_N(\theta)} \ \Pi(dr) \ \Pi(d\beta).
\end{align*}
Trivially, we have $\int_{A_{Nt}} e^{\ell_N(\theta)} \ \Pi(dr) \ \Pi(d\beta) \le \int e^{\ell_N(\theta)} \ \Pi(dr) \ \Pi(d\beta)$. To establish a lower bound, note that for sufficiently large $N$ we have
\begin{align*}
  A_{Nt} \supseteq \{\|\theta - \theta_0\|_N \le N^{-1/4-\delta} / 2\} \cap \{\|\bX b\|_N \le N^{-1/4-\delta} / 2\} \cap \{\|\beta - \mu_\beta\| \le N^{1/4} / 2\} = A^\star_N.
\end{align*}
By (respectively) Assumption~A, the definition of $b$, and Lemma~\ref{lem:2}, each of the components of $A^\star_N$ has a posterior probability of $1 + o_P(1)$. Consequently,
\begin{align*}
  \int_{A_{Nt}} e^{\ell_N(\theta)} \ \Pi(dr) \ \Pi(d\beta)
  &\ge \int_{A^\star_N} e^{\ell_N(\theta)} \ \Pi(dr) \ \Pi(d\beta)
  \\&= \Pi(A^\star_N \mid \bZ_N) \int e^{\ell_N(\theta)} \ \Pi(dr) \ \Pi(d\beta)
  \\&= (1 + o_P(1)) \int e^{\ell_N(\theta)} \ \Pi(dr) \ \Pi(d\beta).
\end{align*}
This completes the proof.

\begin{remark}[On $L_2$-Posterior Concentration]
  The main bottleneck for verifying Assumption~A under the assumptions given is verifying posterior concentration with respect to $\|\theta - \theta_0\|_N$. Standard Bayesian nonparametric arguments for posterior concentration in logistic regression in the fixed-design setting show that appropriately-chosen Gaussian process priors contract with respect to $\|\mu - \mu_0\|^2 = \frac{1}{N} \sum_{i = 1}^N \{\mu(X_i) - \mu_0(X_i)\}^2$ where $\mu(x) = e^{x^\top\beta + r(x)} / (1 + e^{x^\top\beta + r(x)})$; unfortunately, this does not imply concentration with respect to $\|\theta - \theta_0\|_N$ (see, for example, \citealp{szabo2025adaptation}).

  For the sake of concreteness, we give a specific example where we are able to prove the required $L_2$-posterior concentration result via Proposition~\ref{prop:logistic-l2}; the sufficient conditions for this proposition cover truncations of common Gaussian process priors, including the squared exponential kernel and the Matern kernel with regularity $\nu > P / 2$ (see Lemma 3 of \citealp{van2011information}). That is, we can take $\Pi(dr) \propto G(dr) \, 1(\|r\|_\infty \le M)$ to be the conditional distribution under $G$ of $[r \mid \|r\|_\infty \le M]$, where $G(\cdot)$ is a Gaussian process satisfying
  \begin{align*}
    G(\|r\|_\infty \le \epsilon_N) > C_1 e^{-C_2 \, N \epsilon_N^2}.
  \end{align*}
  Proposition~\ref{prop:logistic-l2} also requires a regularity condition on the design $X_1, X_2, \ldots$. This condition is satisfied almost surely if the $X_i$'s are iid from a distribution $F_X$ that has density on $[0,1]^P$ bounded away from $0$.

\end{remark}

\begin{proposition}[A sufficient condition for the $L_2$-concentration part of Assumption~A]
  \label{prop:logistic-l2}
  Consider Model~NC under the assumptions of Theorem~\ref{thm:logistic-bvm}, and suppose that
  $r_0(x) \equiv 0$. Write
  \[
    \theta(x) = x^\top \beta + r(x),
    \qquad
    \theta_0(x) = x^\top \beta_0,
    \qquad
    \mu(x) = F\{\theta(x)\},
    \qquad
    \mu_0(x) = F\{\theta_0(x)\},
  \]
  where $F(t) = (1 + e^{-t})^{-1}$.
  Suppose further that:
  \begin{enumerate}
  \item $\Pi(dr,d\beta) = \Pi_r(dr)\Pi_\beta(d\beta)$, where $\Pi_\beta = \Normal(\mu_\beta,\Sigma_\beta)$ and
    $\Pi_r$ is supported on $\{r : \|r\|_\infty \le M\}$ for some finite $M$;
  \item for some sequence $\epsilon_N \to 0$,
    \[
      \Pi\!\left(
        \|\mu - \mu_0\| > \epsilon_N
        \,\middle|\,
        \bZ_N
      \right) \cinp 0,
      \qquad
      \|\mu - \mu_0\|^2
      =
      \frac{1}{N}\sum_{i=1}^N \{\mu(X_i)-\mu_0(X_i)\}^2;
    \]
  \item there exist constants $u,\rho>0$ such that, for all sufficiently large $N$,
    \[
      \inf_{1\le j\le P}
      \frac{1}{N}\sum_{i=1}^N 1(X_i\in R_j)
      \ge \rho,
    \]
    where
    \[
      R_j
      =
      \Bigl\{
        x\in[0,u]^P :
        x_j \ge \frac{u}{2},
        \quad
        x_k \le \frac{u}{4P}
        \ \text{for all }k\neq j
      \Bigr\}.
    \]
  \end{enumerate}
  Then there exist finite constants $L,A,C>0$ such that
  \[
    \Pi(\|\beta\|_\infty > L \mid \bZ_N) \cinp 0,
    \qquad
    \Pi(\|\theta\|_\infty > A \mid \bZ_N) \cinp 0,
  \]
  and
  \[
    \Pi\!\left(
      \|\theta - \theta_0\|_N > C \epsilon_N
      \,\middle|\,
      \bZ_N
    \right) \cinp 0,
  \]
  where
  \[
    \|\theta - \theta_0\|_N^2
    =
    \frac{1}{N}\sum_{i=1}^N
    \mu_0(X_i)\{1-\mu_0(X_i)\}\{\theta(X_i)-\theta_0(X_i)\}^2.
  \]
  In particular, if $\epsilon_N = N^{-1/4-\delta}$ for some $\delta>0$, then the $L_2$-posterior concentration part of Assumption~A holds.
\end{proposition}

\begin{proof}
  Since $X_i \in [0,1]^P$ and $\theta_0(x)=x^\top\beta_0$, the true linear predictor is uniformly bounded on $[0,1]^P$. Hence there exists a constant $\delta_0 \in (0,1/2)$ such that
  \[
    \delta_0 \le \mu_0(x) \le 1-\delta_0
    \qquad
    \text{for all }x\in[0,1]^P.
  \]
  We first show that the posterior cannot put mass on arbitrarily large values of $\|\beta\|_\infty$.
  Fix $\beta$ with $\|\beta\|_\infty > L$, let $j^\star$ be such that
  $|\beta_{j^\star}| = \|\beta\|_\infty$, and consider $x\in R_{j^\star}$.
  Then
  \[
    x^\top\beta
    =
    x_{j^\star}\beta_{j^\star}
    +
    \sum_{k\neq j^\star} x_k \beta_k.
  \]
  If $\beta_{j^\star}>0$, then
  \[
    x^\top\beta
    \ge
    \frac{u}{2}\|\beta\|_\infty
    -
    \sum_{k\neq j^\star}\frac{u}{4P}\|\beta\|_\infty
    \ge
    \frac{u}{4}\|\beta\|_\infty.
  \]
  If $\beta_{j^\star}<0$, then
  \[
    x^\top\beta
    \le
    -\frac{u}{2}\|\beta\|_\infty
    +
    \sum_{k\neq j^\star}\frac{u}{4P}\|\beta\|_\infty
    \le
    -\frac{u}{4}\|\beta\|_\infty.
  \]
  Since $|r(x)|\le M$ almost surely under $\Pi_r$, it follows that on $R_{j^\star}$ we have either
  \[
    \theta(x) \ge \frac{u}{4}\|\beta\|_\infty - M
    \qquad\text{or}\qquad
    \theta(x) \le -\frac{u}{4}\|\beta\|_\infty + M.
  \]
  Choose $L$ large enough that
  \[
    F\!\left(\frac{uL}{4}-M\right) \ge 1-\frac{\delta_0}{2}
    \qquad\text{and}\qquad
    F\!\left(-\frac{uL}{4}+M\right) \le \frac{\delta_0}{2}.
  \]
  Then on the event $\{\|\beta\|_\infty > L\}$, for every $x\in R_{j^\star}$ we have
  \[
    |\mu(x)-\mu_0(x)| \ge \frac{\delta_0}{2}.
  \]
  By the regularity assumption on the design,
  \[
    \|\mu-\mu_0\|^2
    =
    \frac{1}{N}\sum_{i=1}^N \{\mu(X_i)-\mu_0(X_i)\}^2
    \ge
    \frac{\delta_0^2}{4}\rho
  \]
  for all sufficiently large $N$. Therefore, if $\epsilon_N < \delta_0\sqrt{\rho}/2$, then eventually
  \[
    [\|\mu-\mu_0\|\le \epsilon_N]
    \subseteq
    [\|\beta\|_\infty \le L].
  \]
  Since the posterior contracts in $\|\mu-\mu_0\|$ at rate $\epsilon_N$, it follows that
  \[
    \Pi(\|\beta\|_\infty > L \mid \bZ_N)\cinp 0.
  \]
  Because $|r(x)|\le M$ almost surely and $|x^\top\beta|\le \sqrt{P}\|\beta\|_\infty$ on $[0,1]^P$, we obtain
  \[
    \|\theta\|_\infty \le M + \sqrt{P}\|\beta\|_\infty.
  \]
  Hence, with $A = M + \sqrt{P}L$,
  \[
    \Pi(\|\theta\|_\infty > A \mid \bZ_N)\cinp 0.
  \]
  Finally, work on the event $\{\|\theta\|_\infty \le A\}$, and put
  \[
    A_0 = \max\{A,\|\theta_0\|_\infty\},
    \qquad
    c_A = \inf_{|t|\le A_0} F'(t) > 0.
  \]
  By the mean value theorem, for each $i$ there exists $\xi_i$ between $\theta(X_i)$ and $\theta_0(X_i)$ such that
  \[
    |\mu(X_i)-\mu_0(X_i)|
    =
    F'(\xi_i)\,|\theta(X_i)-\theta_0(X_i)|
    \ge
    c_A\,|\theta(X_i)-\theta_0(X_i)|.
  \]
  Therefore,
  \[
    \|\theta-\theta_0\|_N^2
    =
    \frac{1}{N}\sum_{i=1}^N
    \mu_0(X_i)\{1-\mu_0(X_i)\}\{\theta(X_i)-\theta_0(X_i)\}^2
    \le
    \frac{1}{4c_A^2}\,
    \frac{1}{N}\sum_{i=1}^N \{\mu(X_i)-\mu_0(X_i)\}^2
    =
    \frac{1}{4c_A^2}\|\mu-\mu_0\|^2.
  \]
  Thus
  \[
    \Pi\!\left(
      \|\theta-\theta_0\|_N > \frac{\epsilon_N}{2c_A}
      \,\middle|\,
      \bZ_N
    \right)
    \le
    \Pi(\|\mu-\mu_0\|>\epsilon_N \mid \bZ_N)
    +
    \Pi(\|\theta\|_\infty>A \mid \bZ_N)
    \cinp 0.
  \]
  This proves the result.
\end{proof}

\bibliographystyle{apalike}
\bibliography{references.bib}